\documentclass[12pt]{article} %
\usepackage{amsmath,amssymb,amsthm}
\numberwithin{equation}{section} %
\theoremstyle{plain}
   \newtheorem{thm}{\hspace{\parindent}{\sc Theorem}}[section] %
   \newtheorem{pro}[thm]{\hspace{\parindent}Proposition}
   
   \newtheorem{lem}[thm]{\hspace{\parindent}Lemma}
\theoremstyle{remark} %
   \newtheorem{rem}{\hspace{\parindent}Remark}[section] %
\newcommand{\adots}{a = 0,\pm1,\pm2,\dots}
\newcommand{\Atxrho}{A(t,x;\rho)}

\newcommand{\bR}{\mathbb{R}}

\newcommand{\domain}{[0,T]\times \mathbb{R}^d}
\newcommand{\domains}{[0,T]\times \mathbb{R}^{2d}}
\newcommand{\Eta}{\mathcal{E}_t^0([0,T];}
\newcommand{\Etaprime}{\mathcal{E}_t^0([0,T];B'^a)}
\newcommand{\hdbar}{\hspace{0.08cm}\dbar}
\newcommand{\Hepsilon}{\widetilde{H}_{\epsilon}}
\newcommand{\jdots}{j = 1,2,\dots,d}
\newcommand{\kalphabeta}{k^{(\alpha)}_{\ (\beta)}}

\newcommand{\qepsilon}{q_{\epsilon}}
\newcommand{\Qepsilon}{Q_{\epsilon}}
\newcommand{\rittaire}{\text{\rm{Re\hspace{0.05cm}}}}
\newcommand{\Sspace}{\mathcal{S}(\bR^d)}
\newcommand{\sumjd}{\sum_{j=1}^d}
\newcommand{\txdx}{(t,X,D_x)}
\newcommand{\txxi}{(t,x,\xi)}
\newcommand{\tzdz}{(t,Z,D_z)}
\newcommand{\uepsilon}{u_{\epsilon}}
\newcommand{\utrho}{u(t;\rho)}
\newcommand{\Vtxrho}{V(t,x;\rho)}
\newcommand{\Vktxrho}{V_k(t,x;\rho)}
\newcommand{\Wijtxrho}{W_{ij}(t,x;\rho)}
\newcommand{\weight}{\left(<\xi>^2 + <x>^{2(M+1)}\right)}
\newcommand{\wtrho}{w(t;\rho)}
\newcommand{\zeo}{0 < \epsilon \leq 1}
\def\dbar{{\mathchar'26\mkern-12mud}}
\begin{document}

\title{Notes on the Cauchy problem for the self-adjoint and non-self-adjoint Schr\"odinger equations with polynomially growing potentials}
\author{Wataru Ichinose\thanks{Corresponding author. This research is  partially supported by JSPS KAKENHI grant No. 26400161, JP18K03361 and Shinshu University. E-mail: ichinose@math.shinshu-u.ac.jp
}\ , \and Takayoshi Aoki\thanks{This research is  partially supported by Shinshu University RA grant No.25204513. E-mail: 13sm101h@shinshu-u.ac.jp}
}
\date{}
\maketitle
\begin{quote}
{\small Department of Mathematics, Shinshu University,
Matsumoto 390-8621, Japan}%
\end{quote}\par
\begin{abstract}
\noindent
 The  Cauchy problem is studied for the self-adjoint and non-self-adjoint Schr\"odinger equations. We first prove  the existence and  uniqueness of solutions in the weighted Sobolev spaces.  Secondly we prove  that if potentials are depending continuously and differentiably on a parameter,  so are the solutions, respectively. The non-self-adjoint Schr\"odinger equations that we study are those used  in the theory of continuous  quantum measurements. The results on the existence and  uniqueness of solutions in the weighted Sobolev spaces will play a crucial role in the proof for the convergence of the Feynman path integrals in the theories of quantum mechanics and continuous quantum measurements.
 \end{abstract}
\section{Introduction}
	Let $T > 0$ be an arbitrary constant.
	We will  study the self-adjoint and non-self-adjoint Schr\"odinger equations
\begin{align} \label{1.1.1}
& i\hbar \frac{\partial u}{\partial t}(t)  = \widetilde{H}(t)u(t)  \equiv \Bigl[ H(t) -i\hbar K(t)\Bigr]u(t)\notag\\ 
& := \left[ \frac{1}{2m}\sum_{j=1}^d
      \left(\frac{\hbar}{i}\frac{\partial}{\partial x_j} - qA_j(t,x)\right)^2 + qV(t,x) -i\hbar K(t)\right]u(t),
\end{align}
where $t \in [0,T], x = (x_1,\dotsc,x_d) \in \mathbb{R}^d$, $ \bigl(V(t,x),A(t,x)\bigr) = (V,A_1,A_2,\dotsc,A_d) \in \bR^{d+1}$ are  electromagnetic potentials,  $\hbar$  the Planck constant, $m > 0$  the mass of a particle and $q \in \bR$  its charge. In addition, $K(t)$ is the pseudo-differential operator with a real-valued double symbol $k(t,(x+x')/2,\xi)$ defined by
\begin{equation} \label{1.2.1}
  K\left(t,\frac{X+X'}{2},D_x\right)f  = \iint e^{i(x-y)\cdot\xi}k\left(t,\frac{x+y}{2},\xi\right)f(y)dy\,\dbar\xi
 \end{equation}
for $f \in \Sspace$, where $\dbar\xi = (2\pi)^{-d}d\xi$ and $\Sspace$ is the Schwartz space of all rapidly decreasing functions on $\bR^d$.
 	The non-self-adjoint Schr\"odinger equations that we study in the present paper are those used in the theory of continuous quantum measurements.  See \S 4.2 in \cite{Mensky 1993} and \S 5.1.3 in \cite{Mensky 2000}.  Accordingly, we assume 
\begin{equation} \label{1.3.1}
(K(t)f,f)  \geq -C\Vert f\Vert^2
 \end{equation}
in $[0,T]$ for  $f \in \Sspace$ with a constant $C \geq 0$, where we denote by $L^2 = L^2(\mathbb{R}^d)$   the space of all square integrable functions on
$\mathbb{R}^d$ with inner
product $(f,g) := \int f(x) g(x)^*dx$ for the complex conjugate $g^*$ of $g$ and  norm $\Vert f\Vert$.  For the sake of simplicity we will suppose $\hbar = 1$ and $q = 1$ hereafter. 
\par
In the present paper we consider the potentials $(V,A)$  satisfying
\begin{align} \label{1.4.1}
	&  |\partial^{\alpha}_xV(t,x)| \leq C_{\alpha}<x>,\ |\alpha| \geq 1, \notag \\
&
 \sumjd |\partial^{\alpha}_xA_j(t,x)|
	 \leq C_{\alpha},\ |\alpha| \geq 1
	 \end{align}
 with constants $C_{\alpha} \geq 0$ or 
\begin{equation} \label{1.5.1}
C_0<x>^{2(M+1)} - C_1 
 \leq V(t,x) \leq C_2<x>^{2(M+1)}
 \end{equation}
 with constants $M  > 0, C_0 > 0, C_1 \geq 0$ and $C_2 > 0$ in $\domain$, where $|x| = \left(\sumjd x_j^2\right)^{1/2}$, $<x>  = \left(1 + |x|^2\right)^{1/2}$, $\partial_{x_j} = \partial /\partial
x_j$, $|\alpha| = 
\sum_{j=1}^d
\alpha_j$,  and  $\partial_x^{\alpha} = \partial_{x_1}^{\alpha_1}
\cdots \partial_{x_d}^{\alpha_d}$ for a multi-index
$\alpha = (\alpha_1,\dotsc,\alpha_d)$. 
As well known, if $a $ is a constant greater than 2, the uniqueness of  solutions to \eqref{1.1.1} with $H(t) = -\sum_{j=1}^d\partial_{x_j}^2 - |x|^a$ and  $K(t) = 0$ does not hold (cf. pp. 157-159 in \cite{Berezin et all}, Theorem X.2 in \cite{Reed-Simon II}).  Therefore the assumptions \eqref{1.4.1} and \eqref{1.5.1} are not so restrictive.
\par
For a constant $M \geq 0$   let us introduce the weighted Sobolev spaces 
\begin{align} \label{1.6.1}
B^a_M(\mathbb{R}^d)  := \{f \in & L^2(\mathbb{R}^d);
 \|f\|_{a,M} := \|f\| + \notag\\
& \sum_{|\alpha| \leq  2a} \|\partial_x^{\alpha}f\| +
\|<\cdot>^{2a(M+1)}f\| < \infty\}
 \end{align}
for $a = 1,2,\dotsc$.  We denote the dual space of $B^a_M\ (a = 1,2,\dotsc)$ by $B^{-a}_M$ and the $L^2$ space by $B^0_M$.  \par
  The first aim in the present paper is to prove that for any $u_0 \in B^a_M\ (a =0, \pm1,\pm2,\dots)$ there exists the unique solution $u(t) \in \mathcal{E}^0_t([0,T];B^a_M) \cap \mathcal{E}^1_t([0,T]; \\
  B^{a-1}_M) $ with $u_0$ at $t = 0$ to \eqref{1.1.1}, where ${\cal E}^j_t([0,T];B^a_M)\ (j = 0,1,\dotsc)$ denotes the space of all $B^a_M$-valued, j-times continuously differentiable functions on $[0,T]$. This result will play a crucial role in the proof of the convergence of the Feynman path integrals for \eqref{1.1.1} in \cite{Ichinose 2019} and  \cite{Ichinose 2020} as in the proofs of the theorems in  \cite{Ichinose 1999} and \cite{Ichinose 2014}.
  \par
 
  The second aim in the present paper is to prove that if potentials are depending continuously and differentiably on a parameter,  so  are the solutions to \eqref{1.1.1}   in ${\cal E}^0_t([0,T];B^a_M)\ (a = 0, \pm 1,\pm 2, \cdots)$, respectively.  Such  results have been well known in the theory of ordinary differential equations as the fundamental results.
  	\par
 In the present paper the  results stated above  to \eqref{1.1.1}  will be extended to  multi-particle systems.  For  simplicity we will consider   4-particle systems
\begin{align} \label{1.7.1}
 i\frac{\partial u}{\partial t}(t) & = \widetilde{H}(t)u(t)
 := \Biggl[\sum_{k=1}^4 \Biggl\{\frac{1}{2m_k}
\sum_{j=1}^d
      \left(\frac{1}{i}\frac{\partial}{\partial x_j^{(k)}} - A_j^{(k)}(t,x^{(k)})\right)^2 + 
      \notag\\  
   &   V_k(t,x^{(k)}) -iK_k(t)\Biggr\}
     + \sum_{1\leq i <j \leq 4}W_{ij}(t,x^{(i)}-x^{(j)})\Biggr]u(t) \notag \\
     & \equiv \Biggl[\sum_{k=1}^4\Big\{ H_k(t) -iK_k(t) \Big\}+ \sum_{1\leq i <j \leq 4}W_{ij}(t,x^{(i)}-x^{(j)})\Biggr]u(t),
\end{align}
where $x^{(k)} \in \bR^d\ (k = 1,2,3,4)$ and $K_k(t) = K_k(t,(X^{(k)}+X'^{(k)})/2,D_{x^{(k)}})$ .
\par
	Let's consider the self-adjoint equations, i.e. $K(t) = 0$. When the Hamiltonian $H(t) = \widetilde{H}(t)$ is independent of $t \in [0,T]$,  the existence and  uniqueness of solutions in the $L^2$ space to \eqref{1.1.1} and \eqref{1.7.1}  are equivalent to the self-adjointness of $H(t) = H$ (cf. \S 8.4 in \cite{Reed-Simon I}).  The self-adjointness of $H$ in $L^2$ has almost been settled now (cf. \cite{Cycon et all, Leinfelder et all, Reed-Simon II}).	
	If $H(t)$ is not independent of $t \in [0,T]$, the problem is not simple.  In \cite{Yajima 1991} Yajima has proved the existence and  uniqueness of solutions to \eqref{1.1.1} in $B^a_0\ (a = 0,\pm1,\pm2,\dots)$  under the assumptions
	\begin{align*} 
	&  |\partial^{\alpha}_xV(t,x)| \leq C_{\alpha}, \ |\alpha| \geq 2, \notag\\
	& \sumjd \bigl(|\partial^{\alpha}_xA_j(t,x)| + |\partial^{\alpha}_x\partial_tA_j(t,x)|\bigr)
	 \leq C_{\alpha},\ |\alpha| \geq 1, \notag\\
	 & \sum_{1 \leq j < k \leq d}|\partial^{\alpha}_xB_{jk}(t,x)| \leq C_{\alpha}<x>^{-(1+\delta_{\alpha})},\ |\alpha| \geq 1
	\end{align*}
	with constants $\delta_{\alpha} > 0$ and $C_{\alpha} \geq 0$ by  the theory of Fourier integral operators, where  $B_{jk} =  \partial A_k/\partial x_j  -\partial A_j/\partial x_k. $   In \cite{Ichinose 1995} the first author has proved the existence and  uniqueness of solutions in $B^a_0\ (a = 0,\pm1,\pm2,\dots)$  under the assumptions \eqref{1.4.1}
	by the energy method.  Recently, Yajima in \cite{Yajima 2011,Yajima 2016} has proved  by the semi-group method  the existence and  uniqueness of solutions  in the $L^2$ space to \eqref{1.1.1} and \eqref{1.7.1} with singular potentials.
	\par
	We consider the self-adjoint equations \eqref{1.1.1} and \eqref{1.7.1} again.  When the Hamiltonian $H(t)$ is independent of $t \in [0,T]$, it follows from Theorems VIII. 21 and VIII. 25 in \cite{Reed-Simon I} that if potentials are depending continuously on a parameter,   so  are the solutions in the $L^2$ space.  
	If  $H(t)$ is not independent of $t \in [0,T]$, the problem is not simple again.  In \cite{Ichinose 2012} the first author has proved by the energy method under the assumptions \eqref{1.4.1} that if  potentials  are depending continuously and differentiably on a parameter, so   are the solutions to \eqref{1.1.1}  in $\mathcal{E}_t^0([0,T];B^a_0)$, respectively. \par
	As for the non-self-adjoint Schr\"odinger equations, there are many papers on the spectral analysis (cf. \cite{Davies}, \cite{Sambou}). The authors don't know the results related to our results.
	\par
	Therefore, our aims in the present paper are to generalize the results for the self-adjoint equations \eqref{1.1.1} with potentials \eqref{1.4.1} to   the non-self-adjoint equations \eqref{1.1.1} and \eqref{1.7.1} with potentials \eqref{1.4.1} or \eqref{1.5.1}.
	\par
	We will prove  our results  by the energy method as in \cite{Ichinose 1995} and \cite{Ichinose 2012}.  The crucial point in the proofs of our results for \eqref{1.1.1} is to introduce a family of bounded operators $\big\{\widetilde{H}_{\epsilon}(t)\bigr\}_{0 <\epsilon \leq 1}$ on $B^a_M\ (a = 0,\pm1,\pm2,\dots)$ satisfying Proposition 4.2 in the present paper by (4.1) as an approximation of $\widetilde{H}(t)$.  Then, using the assumption \eqref{1.3.1}, we can complete the proofs  as in \cite{Ichinose 1995} and \cite{Ichinose 2012}.  In the same way the crucial point in the proofs of our results for \eqref{1.7.1} is to introduce
 $\big\{\widetilde{H}_{\epsilon}(t)\bigr\}_{0 <\epsilon \leq 1}$  by \eqref{5.29} 	as an approximation of $\widetilde{H}(t)$.  As in the proofs for \eqref{1.1.1} we can complete the proofs. 
 \par
 	The plan of the present paper is as follows.  In \S 2 we will state all theorems.  \S 3 is devoted to preparing for the proofs of the theorems for \eqref{1.1.1}.  In \S 4 we will prove all theorems for \eqref{1.1.1}.  In \S 5  we will prove all theorems for \eqref{1.7.1}.
\section{Theorems}
In the present  paper we often use symbols $C, C_a, C_{\alpha}, C_{\alpha\beta}$  and $\delta$ to write down constants, though these value are different in general. 
\par
\noindent {\bf Assumption 2.1. }  We assume \eqref{1.3.1}, \eqref{1.4.1} and 
\begin{equation} \label{2.1.1}
|k^{(\alpha)}_{\ (\beta)}(t,x,\xi)| \leq C_{\alpha\beta}(1 + |x| + |\xi|), \  |\alpha + \beta| \geq 1
\end{equation}
in $\domains$, where $k^{(\alpha)}_{\ (\beta)}(t,x,\xi) = (i)^{-|\beta|}\partial_{\xi}^{\alpha}\partial_{x}^{\beta}k(t,x,\xi)$.
\par
\noindent {\bf Assumption 2.2. } Let $M > 0$ be a constant.  We assume \eqref{1.3.1}, \eqref{1.5.1} and 
\begin{align} \label{2.2.1}
& |k_{(\beta)}(t,x,\xi)| \leq C_{\beta}<x>^{M+1},\  |\beta| \geq 1, \notag \\
& |k^{(\alpha)}_{\ (\beta)}(t,x,\xi)| \leq C_{\alpha\beta},\  |\alpha| \geq 1, |\beta| \geq 0.
\end{align}
Suppose for all $\alpha$ and $l = 0,1$ that $\partial_x^{\alpha}\partial_t^lV(t,x)$ and $\partial_x^{\alpha}\partial_t^lA_j(t,x)\ (j = 1,2,\dots,d)$ are continuous in $\domain$ and assume the following.  We have
\begin{equation} \label{2.3.1}
|\partial_x^{\alpha}V(t,x)| \leq C_{\alpha}<x>^{2(M+1)}, |\alpha| \geq 1,
\end{equation}
\begin{equation} \label{2.4.1}
 |\partial_x^{\alpha}\partial_tV(t,x)| \leq C_{\alpha}<x>^{2(M+1)}\end{equation}
for all $\alpha$,
\begin{equation} \label{2.5.1}
 |A_j(t,x)| \leq C<x>^{M+1-\delta}
 \end{equation}
with a constant $\delta > 0$,
\begin{equation} \label{2.6.1}
|\partial_x^{\alpha}A_j(t,x)| \leq C_{\alpha}<x>^{M+1}, |\alpha| \geq 1
 \end{equation}
and 
\begin{equation} \label{2.7.1}
 |\partial_x^{\alpha}\partial_tA_j(t,x)| \leq C_{\alpha}<x>^{M+1}
 \end{equation}
for all $\alpha$.  \vspace{0.3cm}\par
	\begin{thm} (1)
	Suppose Assumption 2.1.  Then,  for any $u_0 \in B^a_0\ (\adots)$ there exists the unique solution $u(t) \in \mathcal{E}_t^0([0,T];B^a_0) \cap \mathcal{E}_t^1([0,T];B^{a-1}_0)$ with $u(0) = u_0$ to \eqref{1.1.1}.  This solution $u(t)$ satisfies 
\begin{equation} \label{2.8.1}
 \Vert u(t) \Vert_{a,0} \leq C_a \Vert u_0 \Vert_{a,0}\quad (0 \leq t \leq T).
 \end{equation}
(2) Suppose Assumption 2.2.  Then we have the same assertions as in (1) where $B^a_0$ is replaced with $B^a_M$.
	\end{thm}
	\begin{rem}
	Let $a(t)$ be a continuous function on $[0,T]$ such that $a(0) = 0$ and $a(t) > 0\ (0 < t \leq T)$.  Since $V := a(t)|x|^4 + |x|^2$ does not satisfy  either \eqref{1.4.1} or \eqref{1.5.1}  for any $M > 0$, Theorem 2.1 can not be applied to \eqref{1.1.1} with $\widetilde{H}(t) := (1/2m)\sumjd (-i\partial_{x_j})^2 +a(t)|x|^4 + |x|^2$.  Theorems 1.2 and 1.4 in \cite{Yajima 2011} can not be applied either, because the self-adjoint operators $\widetilde{H}(t)\ (0 \leq t \leq T)$  don't have a common domain in $L^2(\bR^d)$.
	\end{rem}
	Next, let us consider the  equations \eqref{1.1.1}    depending on a parameter $\rho \in \mathcal{O}$, where $\mathcal{O}$ is an open set in $\bR$.
	%
	\begin{thm}
	We suppose that $\partial_x^{\alpha}\Vtxrho, \partial_x^{\alpha}A_j(t,x;\rho)\ (\jdots)$ and $\kalphabeta(t,x,\xi;\rho)$ are continuous in $[0,T]\times \bR^{2d} \times \mathcal{O} $  for all $\alpha$ and $\beta$. 
	(1) We assume that $(\Vtxrho,\Atxrho)$ and $k(t,x,\xi;\rho)$ satisfy Assumption 2.1   for all $\rho \in \mathcal{O}$ and have the uniform estimates \eqref{1.3.1}, \eqref{1.4.1} and \eqref{2.1.1}  with respect to $\rho \in \mathcal{O}$.
	 Let $u_0 \in B^a_0 \ (\adots)$ be independent of $\rho$ and $u(t;\rho)$
	the solutions to \eqref{1.1.1} with $u(0;\rho) = u_0$ determined in Theorem 2.1.  Then, the mapping : $\mathcal{O} \ni
\rho	\rightarrow \utrho \in \Eta B^a_0)$ is continuous, where the norm in $\Eta B^a_0)$ is $\max_{0 \leq t \leq T}\Vert f(t)\Vert_{a,0}.$  (2) We assume that $(\Vtxrho,\Atxrho)$ and $k(t,x,\xi;\rho)$ satisfy Assumption 2.2  for all $\rho \in \mathcal{O}$ and have the uniform estimates \eqref{2.2.1}-\eqref{2.7.1}  with respect to $\rho \in \mathcal{O}$.
Then we have the same assertions as in (1) where $B^a_0$ is replaced with $B^a_M$.
	\end{thm}

	We set
\begin{equation} \label{2.9.1}
 h(t,x,\xi) :=  \frac{1}{2m}|\xi - A(t,x)|^2 + V(t,x).
 \end{equation}
Then by  \eqref{1.1.1} and \eqref{1.2.1} we have
\begin{equation*} 
 H(t)f =  H\left(t,\frac{X+X'}{2},D_x\right)f  
 \end{equation*}
for $f \in \Sspace.$
\begin{thm} 
 We suppose for $l= 0,1$ that $\partial_{\rho}^l\partial^{\alpha}_x\Vtxrho, \partial_{\rho}^l\partial^{\alpha}_xA_j(t,x;\rho)\ (j= 1,2,\dotsc,d)$ and $\partial_{\rho}^l\kalphabeta (t,
x,
\xi;\rho)$ are continuous in $\domains \times \mathcal{O}$ for all $\alpha$ and $\beta$.
(1)  Besides the assumptions of  (1)  in Theorem 2.2 we assume 
\begin{equation} \label{2.10.1}
\sup_{\rho \in \mathcal{O}} |\partial_{\rho}\partial^{\alpha}_x\Vtxrho| \leq C_{\alpha}<x>^{2},
\end{equation}
\begin{equation} \label{2.11.1}
 \sup_{\rho \in \mathcal{O}} |\partial_{\rho}\partial^{\alpha}_xA_j(t,x;\rho)| \leq C_{\alpha}<x>,
 \end{equation}
\begin{equation} \label{2.12.1}
\sup_{\rho \in \mathcal{O}} |\partial_{\rho}\kalphabeta (t,x,\xi;\rho)| \leq C_{\alpha\beta}(1 + |x|^2 + |\xi|^2)
\end{equation}
in $\domains$ for all $\alpha$ and $\beta$ .  Let $u_0 \in B^{a+1}_0\ (\adots)$ be independent of $\rho$ and $\utrho$ the solutions to \eqref{1.1.1} with $u(0) = u_0$.  Then, the mapping : $ \mathcal{O} \ni \rho \rightarrow \utrho \in \Eta B^a_0)$ is continuously differentiable with respect to $\rho$, we have
\begin{equation} \label{2.13.1}
\sup_{\rho \in \mathcal{O}} \Vert \partial_{\rho}\utrho\Vert_{a,0} \leq C_a \Vert u_0 \Vert_{a+1,0}\quad (0 \leq t \leq T)
 \end{equation}
and $\partial_{\rho}\utrho$ is the solution to
\begin{equation} \label{2.14.1}
 i\frac{\partial}{\partial t}\wtrho = \widetilde{H}(t;\rho)\wtrho + \frac{\partial  \widetilde{H}(t;\rho)}{\partial\rho}\utrho
  \end{equation}
with $w(0) = 0$. Here, $\partial_{\rho} \widetilde{H}(t;\rho)$ denotes the pseudo-differential operator with the double symbol $\partial_{\rho} \widetilde{h}(t,(x+x')/2,\xi;\rho)$, where  $\widetilde{h}(t,x,\xi;\rho) = h(t,x,\xi;\rho) - ik(t,x,\xi;\rho)$. (2)  Besides the assumptions of  (2)  in Theorem 2.2 we assume
\begin{equation} \label{2.15.1}
\sup_{\rho \in \mathcal{O}} |\partial_{\rho}\partial^{\alpha}_x\Vtxrho| \leq C_{\alpha}<x>^{2(M+1)},
\end{equation}
\begin{equation} \label{2.16.1}
 \sup_{\rho \in \mathcal{O}} |\partial_{\rho}\partial^{\alpha}_xA_j(t,x;\rho)| \leq C_{\alpha}<x>^{M+1},
 \end{equation}
\begin{equation} \label{2.17.1}
\sup_{\rho \in \mathcal{O}} |\partial_{\rho}\kalphabeta (t,x,\xi;\rho)| \leq C_{\alpha\beta}(<x>^{2(M+1)}+ <\xi>^2)
\end{equation}
in $\domains$ for all $\alpha$ and $\beta$.  Then we have the same assertions  as in (1) where $B^a_0$ is replaced with $B^a_M$.
\end{thm}
	Now, we consider the 4-particle systems \eqref{1.7.1}.\par
	{\bf Assumption 2.3.}  (1) Each $(V_k(t,x),A^{(k)}(t,x))$ and $K_k(t) \ (k = 1,2)$ satisfies Assumption 2.2 with $M = M_k > 0$. (2)  Each $(V_k,A^{(k)})$ and $K_k(t)\ (k = 3,4)$ satisfies Assumption 2.1.
	(3)  For $M_0 := \min(M_1,M_2)$   $W_{12}$ satisfies
\begin{equation} \label{2.18.1}
|W_{12}(t,x)| \leq C<x>^{2(M_0 +1)-\delta}
  \end{equation}
with a constant $\delta > 0$ and 
\begin{equation} \label{2.19.1}
|\partial_x^{\alpha}W_{12}(t,x)| \leq C_{\alpha}<x>^{2(M_0 +1)},\  |\alpha| \geq 1.
  \end{equation}
(4) Each $W_{ij}(t,x)$ except $W_{12}$ satisfies 
\begin{equation} \label{2.20.1}
|\partial_x^{\alpha}W_{ij}(t,x)| \leq C_{\alpha}<x>,\  |\alpha| \geq 1.
  \end{equation}
\par
	We introduce the weighted Sobolev spaces $
B'^a(\mathbb{R}^{4d})  := \{f \in  L^2(\mathbb{R}^{4d});
 \|f\|_a := \|f\| + 
 \sum_{|\alpha| \leq  2a} \|\partial_x^{\alpha}f\| + \sum_{k=1}^4
\|<x^{(k)}>^{2a(M_k+1)}f\| < \infty\}\ (a = 1,2,\dots)
$ with $M_3 = M_4 = 0$. We denote the dual space of $B'^a$ by $B'^{-a}$ and $L^2$ by $B'^0$. 
\begin{thm}
Under Assumption 2.3 for any $u_0 \in B'^a(\bR^{4d})\ (a = 0,\pm1,\pm2,\\
\dots)$ there exists the unique solution $u(t) \in \mathcal{E}_t^0([0,T];B'^a) \cap \mathcal{E}_t^1([0,T];B'^{a-1})$ with $u(0) = u_0$ to \eqref{1.7.1}.  This solution $u(t)$ satisfies 
\begin{equation} \label{2.21.1}
 \Vert u(t) \Vert_a \leq C_a \Vert u_0 \Vert_a\quad (0 \leq t \leq T).
 \end{equation}
\end{thm}
	We will consider the 4-particle systems \eqref{1.7.1}  depending on a parameter $\rho \in \mathcal{O}$.
\begin{thm}
We suppose that $\partial_x^{\alpha}V_k(t,x;\rho)$,  $\partial_x^{\alpha}A^{(k)}_j(t,x;\rho) \ (k = 1,2,3,4, \\
\jdots)$, 
$k^{(\alpha)}_{\/l \ (\beta)}(t,x,\xi;\rho)\ (l = 1,2,3,4)$ and  
$\partial_x^{\alpha}W_{ij}(t,x;\rho)\ (1 \leq i < j \leq 4)$ are continuous in $[0,T]\times \bR^{2d} \times \mathcal{O}$
 for all $\alpha$ and $\beta$.
In addition, we assume 
that $(V_k(t,x;\rho),A^{(k)}(t,x;\rho))$ and $K_k(t;\rho) \ (k= 1,2,3,4)$ and $\Wijtxrho\ (1 \leq i < j \leq 4)$ satisfy Assumption 2.3  for all $\rho \in \mathcal{O}$ and have the uniform estimates   with respect to $\rho \in \mathcal{O}$ for all inequalities stated in Assumption 2.3.  \par
Let $u_0 \in B'^a \ (\adots)$ be independent of $\rho$ and $u(t;\rho)$
	the solutions to \eqref{1.7.1} with $u(0;\rho) = u_0$ determined in Theorem 2.4.  Then, the mapping : $\mathcal{O} \ni
\rho	\rightarrow \utrho \in \Etaprime$ is continuous. 
\end{thm}
\begin{thm}
Besides the assumptions of Theorem 2.5 we suppose for all $\alpha$ and $\beta$ that all functions $\partial_{\rho}\partial^{\alpha}_x\Vktxrho, \partial_{\rho}\partial^{\alpha}_xA^{(k)}_j(t,x;\rho),  \partial_{\rho}k^{(\alpha)}_{\/l \ (\beta)}(t,x,\xi;\rho)$ and $\partial_{\rho}\partial^{\alpha}_xW_{ij}(t,x;\rho)$ are continuous in $\domains \times \mathcal{O}$.  In addition, we assume \eqref{2.15.1} - \eqref{2.17.1}  with $M = M_k$ for $(V_k,A^{(k)})$ and $K_k(t)\ (k = 1,2)$, \eqref{2.10.1} - \eqref{2.12.1} for $(V_k,A^{(k)})$ and $K_k(t)\ (k = 3,4)$,
\begin{equation} \label{2.22.1}
\sup_{\rho \in \mathcal{O}} |\partial_{\rho}\partial^{\alpha}_xW_{12}(t,x;\rho)| \leq C_{\alpha}<x>^{2(M_0+1)}
\end{equation}
for all $\alpha$ and 
\begin{equation} \label{2.23.1}
\sup_{\rho \in \mathcal{O}} |\partial_{\rho}\partial^{\alpha}_xW_{ij}(t,x;\rho)| \leq C_{\alpha}<x>^2, \ (i,j) \not= (1,2)
\end{equation}
for all $\alpha$.
\par
 Let $u_0 \in B'^{a+1}\ (\adots)$ be independent of $\rho$ and $\utrho$ the solutions to \eqref{1.7.1} with $u(0;\rho) = u_0$.  Then we have the same assertion as in Theorem 2.3.
 \end{thm}
 %
 \section{Preliminaries}
 Let $h(t,x,\xi)$ be the function defined by \eqref{2.9.1}.
 \begin{lem}
 Assume \eqref{1.5.1} and \eqref{2.5.1}.  Then there exist constants $C_0^* > 0$ and $C_1^*\geq 0$ such that 
 \begin{equation} \label{3.1}
 C_0^*(<\xi>^2 +  <x>^{2(M+1)}) - C_1^* \leq h(t,x,\xi)  \leq  C_0^{*-1}(<\xi>^2 +  <x>^{2(M+1)})
 \end{equation}
 in $\domains$.
 \end{lem}
 \begin{proof}
 From \eqref{2.9.1} we have $h(t,x,\xi) \leq (|\xi|^2 + |A(t,x)|^2)/m + V(t,x)$ and hence by \eqref{1.5.1} and \eqref{2.5.1}
 \begin{equation*}
 h(t,x,\xi) \leq C(<\xi>^2 + <x>^{2(M+1)})
 \end{equation*}
 in $\domains$ with a constant $C \geq 0$. \par
   We may assume $0 < \delta \leq M+1$ in \eqref{2.5.1}.  Take $p > 1$ and $q > 1$ so that 
 \begin{equation*}
\frac{1}{p} = \frac{1}{2}\left(1 - \frac{1}{2}\cdot\frac{\delta}{M+1}\right),\quad \frac{1}{p} + \frac{1}{q} = 1.
 \end{equation*}
Then we have
 \begin{align*}
& p(M+1 - \delta) = 2(M+1)\cdot\frac{1-\frac{\delta}{M+1}}{1-\frac{1}{2}\cdot\frac{\delta}{M+1}} \equiv 2(M+1)\delta_1,
\notag \\
&q = \frac{2}{1+\frac{1}{2}\cdot\frac{\delta}{M+1}} \equiv 2\delta_2
 \end{align*}
 with $0 < \delta_j < 1\ (j =  1,2).$  Hence, Young's inequality and \eqref{2.5.1} show
 \begin{align*}
&  |A(t,x)|\cdot|\xi| \leq \frac{1}{p}|A|^p + \frac{1}{q}|\xi|^q \leq \frac{1}{p}<x>^{p(M+1-\delta)} + \frac{1}{q}|\xi|^q
\notag \\
&= \frac{1}{p}<x>^{2(M+1)\delta_1} + \frac{1}{q}|\xi|^{2\delta_2}.
\end{align*}
 Applying this, \eqref{1.5.1} and \eqref{2.5.1} to \eqref{2.9.1}, we have
 \begin{align*} 
  h(t,x,\xi)  & \geq  \frac{1}{2m}\left(|\xi|^2  - 2|A| \cdot |\xi|\right) + V \\
   & \geq C_0(<\xi>^2 - <x>^{2(M+1)\delta_1} -<\xi>^{2\delta_2} + <x>^{2(M+1)}) - C_1
 \end{align*}
 with constants $C_0 > 0$ and $C_1 \geq 0$.  Therefore, we obtain \eqref{3.1}.
 \end{proof}
 	We fix $C_0^*$ and $C_1^*$ in Lemma 3.1 hereafter.  We set 
	\begin{equation} \label{3.2}
	h_s(t,x,\xi) = h(t,x,\xi) + \frac{i}{2m}\nabla\cdot A(t,x),
	\end{equation}
	where $\nabla\cdot A(t,x) = \sum_{j=1}^d\partial_{x_j}A_j(t,x).$
	Since  the real part $\rittaire h_s(t,x,\xi)$ of $h_s(t,x,\xi)$ is equal to $h(t,x,\xi)$, we can determine
	\begin{equation} \label{3.3}
	p_{\mu}(t,x,\xi) := \frac{1}{\mu + h_s(t,x,\xi)}
	\end{equation}
	for $\mu \geq C_1^*$ under the assumptions of Lemma 3.1.  We  denote by $H_s(t,X,D_x)f$ the pseudo-differential operator 
	\begin{equation*} 
	\int e^{ix\cdot\xi}h_s(t,x,\xi)\hdbar\xi\int e^{-iy\cdot\xi}f(y)dy
	\end{equation*}
	for $f \in \Sspace$ with the symbol $h_s(t,x,\xi)$.  As is well known (cf. Theorem 2.5 in Chapter 2 of \cite{Kumano-go}), $H_s(t,X,D_x) = H(t)$ holds, where $H(t)$ is the operator defined by \eqref{1.1.1}.
	\begin{lem}
	Assume \eqref{1.5.1}, \eqref{2.3.1} and \eqref{2.5.1} - \eqref{2.6.1}.  Then we have
	\begin{equation} \label{3.4}
	\bigl[\mu + H(t)\bigr]P_{\mu}\txdx = I + R_{\mu}\txdx,
	\end{equation}
	\begin{equation} \label{3.5}
	\left|r^{(\alpha)}_{\mu\ (\beta)}(t,x,\xi)\right| \leq  C_{\alpha\beta}\left(\mu - C_1^*\right)^{-1/2}
	\end{equation}
	in $\domains$ for $\mu \geq C_0^*/2 + C_1^*$ with constants $C_{\alpha\beta}$ independent of $\mu$. 
	\end{lem}
	\begin{proof}
	Let $\mu \geq C_1^*$.
	By Lemma 3.1 and \eqref{3.2} we have
 \begin{equation} \label{3.6}
  C_0^*(<\xi>^2  + <x>^{2(M+1)}) + \mu - C_1^* \leq   \mu + \rittaire h_s(t,x,\xi).
   \end{equation}
     Since $H(t) = H_s\txdx$, from (2.13) in \cite{Ichinose 1995} we have
   \begin{align} \label{3.7}
  & r_{\mu}(t,x,\xi) = \sum_{|\alpha| = 1} \int_0^1 d\theta\  \text{Os}-\iint e^{-iy\cdot\eta}h_s^{(\alpha)}(t,x,\xi+\theta\eta)
  p_{\mu (\alpha)}(t,x+y,\xi)dy\hdbar\eta  \notag\\
 &  =  \sum_{|\alpha| = 1} \int_0^1 d\theta\  \text{Os}-\iint e^{-iy\cdot\eta}
 <y>^{-2l_0}  <D_{\eta}>^{2l_0}<\eta>^{-2l_1} <D_y>^{2l_1} \notag\\
 & \hspace{2cm} \cdot h_s^{(\alpha)}(t,x,\xi+\theta\eta)  p_{\mu (\alpha)}(t,x+y,\xi)dy\hdbar\eta
 \end{align}
  for large integers $l_0$ and $l_1$, where $<D_\eta>^2 = 1 - \sumjd \partial_{\eta_j}^2.$ \par
  	Now, using \eqref{2.3.1} and \eqref{2.5.1} - \eqref{2.6.1}, from \eqref{3.2} we have
   \begin{align*} 
  & |\partial_{x_j} h_s(t,x+y,\xi)| \leq C\left(<\xi><x+y>^{M+1} + <x + y>^{2(M+1)}\right) \notag\\
 &   \leq C'\left(<\xi>^2 + <x + y>^{2(M+1)}\right).
 \end{align*}
 In the same way we can prove
 \begin{equation} \label{3.8}
 |h^{(\alpha)}_{s\ (\beta)}(t,x+y,\xi)| \leq  C \left(<\xi>^2 + <x+y >^{2(M+1)}\right) 
 \end{equation}
 for all $\alpha$ and $|\beta|\geq 1$, and 
 \begin{equation} \label{3.9}
 |h^{(\alpha)}_{s\ (\beta)}(t,x,\xi + \theta\eta)| \leq  C \left(<\xi>+ <x >^{M+1}\right) <\eta>
 \end{equation}
 for $|\alpha|\geq 1$ and all $\beta$.  We also note
   \begin{align*} 
  & \frac{1}{<\xi>^2 + <x+y>^{2(M+1)}} \leq  \frac{1}{<\xi>^2 + <x>^{2(M+1)}/(\sqrt{2}<y>)^{2(M+1)}}\notag\\
 &   \leq \frac{(\sqrt{2}<y>)^{2(M+1)}}{<\xi>^2 + <x>^{2(M+1)}}.
 \end{align*}
  Apply \eqref{3.6} and \eqref{3.8} - \eqref{3.9} to \eqref{3.7}.  Then, 
  taking integers $l_0$ and $l_1$ so that $2l_0- 2(M+1) > d$ and $2l_1 - 1 > d$, we have 
   \begin{align} \label{3.10}
  & |r_{\mu}(t,x,\xi)| \leq C \iint
 <y>^{-2l_0} <\eta>^{-2l_1}<\eta> <y>^{2(M+1)}dy\hdbar\eta \notag\\
 & \hspace{1cm} \times \frac{\Theta^{1/2}}{\mu - C_1^* + C_0^*\Theta}	\leq C' \max_{1 \leq \theta}  \frac{\theta^{1/2}}{\mu - C_1^* + C_0^*\theta}
 \end{align}
 with constants $C$ and $C'$ independent of $\mu \geq C_1^*$, where  $\Theta = <\xi>^2 + \\
 <x>^{2(M+1)}$.
 Applying (2.9) in \cite{Ichinose 1995} with $\kappa = 1$ and $\tau = 2$ to \eqref{3.10}, we have
 \begin{equation*} 
  |r_{\mu}(t,x,\xi)| \leq C_0(\mu - C_1^*)^{-1/2} 
  \end{equation*}
 for $\mu \geq C_0^*/2 + C_1^*.$  In the same way we can prove \eqref{3.5} from \eqref{3.7} - \eqref{3.9}.
	\end{proof}
	\begin{pro}
	Under the assumptions of Lemma 3.2 there exist a constant $\mu \geq C_0^*/2 + C_1^*$ and a function $w(t,x,\xi)$ in $\domains$ satisfying 
 \begin{equation} \label{3.11}
  |w^{(\alpha)}_{\ \  (\beta)}(t,x,\xi)| \leq C_{\alpha\beta}\left(<\xi>^2 + <x>^{2(M+1)}\right)^{-1}
  \end{equation}
  for all $\alpha, \beta$ and 
 \begin{equation} \label{3.12}
 W\txdx = \bigl(\mu+ H(t)\bigr)^{-1}.
   \end{equation}
	\end{pro}
	\begin{proof} Let $\mu \geq C_0^*/2 + C_1^*$.
	From \eqref{3.3}, \eqref{3.6} and \eqref{3.8} - \eqref{3.9} we see
 \begin{equation*} 
  |p^{(\alpha)}_{\mu\  (\beta)}(t,x,\xi)| \leq C_{\alpha\beta}\left(<\xi>^2 + <x>^{2(M+1)}\right)^{-1}
  \end{equation*}
    for all $\alpha$ and $ \beta$.  Hence we can complete the proof of Proposition 3.3 as in the proof of (2.16) of \cite{Ichinose 1995} by using Lemma 3.2.
	\end{proof}
	We take a constant $\mu \geq C_0^*/2 + C_1^*$  stated in Proposition 3.3 and fix it hereafter throughout \S 3 and \S4.  Set
 \begin{equation} \label{3.13}
\lambda(t,x,\xi) := \mu + h_s(t,x,\xi).
   \end{equation}
 Then from \eqref{3.2} we have
 \begin{equation} \label{3.14}
\Lambda\txdx = \mu + H_s(t,X,D_x) = \mu + H(t).
   \end{equation}
We take a $\chi \in \Sspace$ such that $\chi(0) = 1$ and set 
 \begin{equation} \label{3.15}
\chi_{\epsilon}(t,x,\xi) := \chi\bigl(\epsilon(\mu + h(t,x,\xi)\bigr)
   \end{equation}
 for constants $0 < \epsilon \leq 1$.  We note that $h(t,x,\xi)$ defined by \eqref{2.9.1} is real-valued. \par
 	The following is crucial in the present paper.
	\begin{lem}
	Under the assumptions of Lemma 3.2 there exist functions $\omega_{\epsilon}(t,x,\xi)\ (\zeo)$ in $\domains$ satisfying
 \begin{equation} \label{3.16}
 \sup_{0 < \epsilon \leq 1}\sup_{t,x,\xi} |\omega^{(\alpha)}_{\epsilon\  (\beta)}(t,x,\xi)| \leq C_{\alpha\beta} < \infty
  \end{equation}
  for all $\alpha, \beta$ and 
 \begin{equation} \label{3.17}
\Omega_{\epsilon}\txdx = \Bigl[X_{\epsilon}\txdx, \Lambda\txdx\Bigr],
   \end{equation}
where $[\cdot,\cdot]$ denotes the commutator of operators.
	\end{lem}
	\begin{proof}
	Apply Theorem 3.1 in Chapter 2 of \cite{Kumano-go} to the right-hand side of \eqref{3.17}.  Then we have
   \begin{align} \label{3.18}
  & \omega_{\epsilon}(t,x,\xi) = \sum_{|\alpha| = 1} \Bigl\{\chi_{\epsilon}^{(\alpha)}(t,x,\xi)\lambda_{(\alpha)}(t,x,\xi)  
  - \lambda^{(\alpha)}(t,x,\xi)\chi_{\epsilon(\alpha)}(t,x,\xi) \Bigr\}\notag\\
 &\quad   + 2 \sum_{|\gamma| = 2}\frac{1}{\gamma\, !}\int_0^1 (1 - \theta)d\theta\  \text{Os}-\iint e^{-iy\cdot\eta}
 \Bigl\{\chi_{\epsilon}^{(\gamma)}(t,x,\xi+\theta\eta)\lambda_{(\gamma)}(t,x+y,\xi)  \notag\\
 & \qquad \quad 
  - \lambda^{(\gamma)}(t,x,\xi+\theta\eta)\chi_{\epsilon(\gamma)}(t,x+y,\xi) \Bigr\}dy\hdbar\eta 
  \equiv I_{1\epsilon} + I_{2\epsilon}.
 \end{align}
 By \eqref{3.2}, \eqref{3.13} and \eqref{3.15} we can write 
   \begin{align} \label{3.19}
  & I_{1\epsilon}(t,x,\xi) = \epsilon\chi\,'(\epsilon(\mu +h))\sum_{|\alpha| = 1}\Bigl\{h^{(\alpha)}h_{s(\alpha)}
  - h^{(\alpha)}_sh_{(\alpha)} \Bigr\}\notag\\
 & =  \epsilon\chi\,'\bigl(\epsilon(\mu +h(t,x,\xi))\bigr)\sum_{|\alpha| = 1}\frac{i}{2m^2}(\xi_{\alpha} - A_{\alpha}(t,x))(-i\partial_x)^{\alpha}\nabla\cdot A(t,x).
 \end{align}
 Hence, using $\epsilon\chi\,'(\epsilon(\mu +h))
 = (\mu +h)^{-1}\bigl\{\epsilon(\mu +h)\chi \,'(\epsilon(\mu +h))\bigr\}$ and Lemma 3.1, from \eqref{2.5.1} - \eqref{2.6.1} we can prove
 $
\sup_{\zeo}\sup_{t,x,\xi}|I_{1\epsilon}| < \infty. 
$
 In the same way from \eqref{3.19} we can prove
 \begin{equation} \label{3.20}
 \sup_{0 < \epsilon \leq 1}\sup_{t,x,\xi} |I^{(\alpha)}_{1\epsilon\  (\beta)}(t,x,\xi)| \leq C_{\alpha\beta} < \infty
  \end{equation}
  for all $\alpha$ and $\beta$.
  \par
	Next we will consider $I_{2\epsilon}$. Let $|\gamma| = 2$.  Since from \eqref{3.15} we have 
 \begin{equation*} 
\partial_{\xi_j}\chi_{\epsilon}(t,x,\xi) = \frac{1}{m}\epsilon \chi \,'(\epsilon(\mu +h))(\xi_j - A_j)
  \end{equation*}
  and
 \begin{equation*} 
\epsilon(\xi_j - A_j)\partial_{\xi_k}\chi \,'(\epsilon(\mu +h)) =  \frac{1}{m}\epsilon^2(\xi_j - A_j)(\xi_k - A_k)\chi \,''(\epsilon(\mu +h)),
  \end{equation*}
  as in the proof of \eqref{3.20} we can easily prove
 \begin{equation} \label{3.21}
\sup_{\zeo}  |\chi^{(\alpha+\gamma)}_{\epsilon\ \ (\beta)}(t,x,\xi)| \leq C_{\alpha\beta}\left(<\xi>^2 + <x>^{2(M+1)}\right)^{-1}
  \end{equation}
  for all $\alpha$ and $\beta$.  In the same way we can also prove
 \begin{equation} \label{3.22}
\sup_{\zeo}  |\chi^{(\alpha)}_{\epsilon\ (\beta+\gamma)}(t,x,\xi)| \leq C_{\alpha\beta} < \infty
  \end{equation}
   for all $\alpha$ and $\beta$.  We also note from \eqref{3.13} that each of $\lambda^{(\gamma)} = h_s^{(\gamma)}$ for $|\gamma| = 2$ is equal to $0$ or $1/m$.
   Hence, applying \eqref{3.8} and \eqref{3.21} - \eqref{3.22} to $I_{2\epsilon}$ in \eqref{3.18}, as in the proof of \eqref{3.10} we have 
 $
\sup_{\zeo}\sup_{t,x,\xi}|I_{2\epsilon}| < \infty. 
$
In the same way we can prove 
 \begin{equation*} 
 \sup_{0 < \epsilon \leq 1}\sup_{t,x,\xi} |I^{(\alpha)}_{2\epsilon\  (\beta)}(t,x,\xi)| \leq C_{\alpha\beta} < \infty
  \end{equation*}
   for all $\alpha$ and $\beta$, which 
   completes the proof of Lemma 3.4 together with \eqref{3.20}.
	\end{proof}
	Let 
 \begin{equation} \label{3.23}
\lambda_{M}(x,\xi) := \mu\,' + \frac{1}{2m}|\xi|^2 + <x>^{2(M+1)},
  \end{equation}
  which is equal to $\lambda(t,x,\xi)$ defined by \eqref{3.13} with $V = \, <x>^{2(M+1)}$ and $A = 0$.  Let $B^a_M$ 
be the weighted Sobolev spaces introduced in \S 1.
	\begin{pro}
	(1)  There exist a constant $\mu\,' \geq 0$  and a function $w_M(x,\xi)$ in $\domains$ satisfying 
	\eqref{3.11}  
  for all $\alpha, \beta$ and 
 \begin{equation} \label{3.24}
 W_M(X,D_x) = \Lambda_M(X,D_x)^{-1}.
   \end{equation}
(2) We take a $\mu\,' \geq 0$ satisfying (1).  Let $f$ be in the dual space $\mathcal{S}\,'(\bR^d)$ of $\Sspace$.  Then, $B^a_M \ni f \ (\adots)$ is equivalent to $ (\Lambda_M)^a f \in L^2.$
	\end{pro}
	\begin{proof}
	The assertion (1) follows from Proposition 3.3.  The assertion (2) follows from Lemma 2.4  of \cite{Ichinose 1995} with $s = a, a = 2(M+1)$ and $b=2$.
	\end{proof}
	We take a constant $\mu\,' \geq 0$ stated in Proposition 3.5 and fix it hereafter throughout \S 3 and \S 4.  We can easily see from \eqref{3.1}, \eqref{3.8} and \eqref{3.9} that under the assumptions of Lemma 3.2  we have
 \begin{equation} \label{3.25}
|h^{(\alpha)}_{s\  (\beta)}(t,x,\xi)| \leq C_{\alpha\beta}\weight
  \end{equation}
  in $\domains$ for all $\alpha$ and $\beta$.
  %
	\section{Proofs of Theorems 2.1 - 2.3}
	Let $\lambda(t,x,\xi)$ and $\chi_{\epsilon}\txxi\ (\zeo)$ be the functions defined by \eqref{3.13} and \eqref{3.15}, respectively.  We define an approximation of $\widetilde{H}(t)$ by the product of operators
 \begin{equation} \label{4.1}
\widetilde{H}_{\epsilon}(t) := X_{\epsilon}\txdx^{\dagger}\widetilde{H}(t)X_{\epsilon}\txdx,
   \end{equation}
 where $X_{\epsilon}\txdx^{\dagger}$ denotes the formally adjoint operator of $X_{\epsilon}\txdx$.
 \begin{lem}
 Under Assumption 2.2 there exist functions $\qepsilon\txxi\ (\zeo)$ satisfying
 \begin{equation} \label{4.2}
 \sup_{0 < \epsilon \leq 1}\sup_{t,x,\xi} |q^{(\alpha)}_{\epsilon\  (\beta)}(t,x,\xi)| \leq C_{\alpha\beta} < \infty
  \end{equation}
  for all $\alpha, \beta$ and
 \begin{align} \label{4.3}
Q_{\epsilon}\txdx = &\Bigl[\Lambda\txdx, \widetilde{H}_{\epsilon}(t)\Bigr]\Lambda\txdx^{-1} \notag \\
&  + i\frac{\partial\Lambda}{\partial t}\txdx\Lambda\txdx^{-1}.
   \end{align}
 \end{lem}
 \begin{proof}
 We first note 
 \begin{align*}
 & \Bigl[\Lambda\txdx, \widetilde{H}_{\epsilon}(t)\Bigr] = \Bigl[\Lambda(t), X_{\epsilon}(t)^{\dagger}\Bigr]\widetilde{H}(t)X_{\epsilon}(t) + X_{\epsilon}(t)^{\dagger}\Bigl[\Lambda(t), H(t)\Bigr]X_{\epsilon}(t)\\
 &\quad  - i X_{\epsilon}(t)^{\dagger}\Bigl[\Lambda(t), K(t)\Bigr]X_{\epsilon}(t) + 
 X_{\epsilon}(t)^{\dagger}\widetilde{H}(t)\Bigl[\Lambda(t), X_{\epsilon}(t)\Bigr].
 \end{align*}
 Since $\Lambda(t) = \mu + H(t)$ from \eqref{3.14}, we have 
 $[\Lambda(t),H(t)] = 0$ and  $\Lambda(t)^{\dagger} = \Lambda(t)$. Hence
 \begin{align*}
  \Bigl[\Lambda(t), \widetilde{H}_{\epsilon}(t)\Bigr] & = -\Bigl[\Lambda(t), X_{\epsilon}(t)\Bigr]^{\dagger}\widetilde{H}(t)X_{\epsilon}(t) +  X_{\epsilon}(t)^{\dagger}\widetilde{H}(t)\Bigl[\Lambda(t), X_{\epsilon}(t)\Bigr] \notag \\
  & - i X_{\epsilon}(t)^{\dagger}\Bigl[\Lambda(t), K(t)\Bigr]X_{\epsilon}(t).
   \end{align*}
   As in the proof of Lemma 3.4, we can prove from Assumption 2.2 and Proposition 3.3 that there exist $\qepsilon '(t,x,\xi) \ (0 < \epsilon \leq 1)$ satisfying \eqref{4.2} and 
   \begin{equation*}
   \Qepsilon '(t,X,D_x) = X_{\epsilon}(t)^{\dagger}\Bigl[\Lambda(t), K(t)\Bigr]X_{\epsilon}(t)\Lambda(t)^{-1} .
   \end{equation*}
 From \eqref{2.2.1} we have
   \begin{equation}  \label{4.4}
   |k^{(\alpha)}_{\ (\beta)}(t,x,\xi)| \leq C_{\alpha\beta}( <\xi>^2 + <x>^{2(M+1)})
      \end{equation}
for all $\alpha$ and $\beta$.  Consequently, using  Proposition 3.3 and Lemma 3.4, we see together with \eqref{3.25} that there exist functions $q_{\epsilon}''\txxi\ (\zeo)$ satisfying \eqref{4.2} and 
 \begin{align*} 
Q_{\epsilon}''\txdx = &\Bigl[\Lambda\txdx, \widetilde{H}_{\epsilon}(t)\Bigr]\Lambda\txdx^{-1}.
 \end{align*}
 It is easy to study the second term on the right-hand side of  \eqref{4.3} by Proposition 3.3. Thus our proof is complete.
 \end{proof}
 \begin{pro}
 Under Assumption 2.2 there exist functions $q_{a\epsilon}\txxi\ (\adots, \zeo)$ satisfying \eqref{4.2} and 
 \begin{equation} \label{4.5}
Q_{a\epsilon}\txdx = \Biggl[i\frac{\partial}{\partial t}- \widetilde{H}_{\epsilon}(t), \Lambda(t)^a\Biggr]\Lambda(t)^{-a}.
   \end{equation}
 \end{pro}
 \begin{proof}
 For $a = 0$ the assertion is clear.  For $a = 1$ the assertion follows from Lemma 4.1. Consider the case $a = 2$.  We note
 \begin{equation} \label{4.6}
[P,QR] = [P,Q]R + Q[P,R]
   \end{equation}
and thereby
 \begin{align*} 
 &\Biggl[i\frac{\partial}{\partial t}- \widetilde{H}_{\epsilon}(t), \Lambda(t)^2\Biggr]\Lambda(t)^{-2} = \Biggl[i\frac{\partial}{\partial t}- \widetilde{H}_{\epsilon}(t), \Lambda(t)\Biggr]\Lambda(t)^{-1} \\
 & + \Lambda(t)\Biggl[i\frac{\partial}{\partial t}- \widetilde{H}_{\epsilon}(t), \Lambda(t)\Biggr]\Lambda(t)^{-2} 
 =  \Qepsilon(t) + \Lambda(t)\Qepsilon(t)\Lambda(t)^{-1}.
 \end{align*}
 Hence it follows from Lemma 4.1 and Proposition 3.3 that the assertion holds.  We consider the case $a = 3$.  From \eqref{4.6} we have 
 \begin{align*} 
 &\Biggl[i\frac{\partial}{\partial t}- \widetilde{H}_{\epsilon}(t), \Lambda(t)^3\Biggr]\Lambda(t)^{-3} = \Biggl[i\frac{\partial}{\partial t}- \widetilde{H}_{\epsilon}(t), \Lambda(t)\Biggr]\Lambda(t)^{-1} \\
 & + \Lambda(t)\Biggl\{\Biggl[i\frac{\partial}{\partial t}- \widetilde{H}_{\epsilon}(t), \Lambda(t)^2\Biggr]\Lambda(t)^{-2}\Bigg\} \Lambda(t)^{-1} \\
 & = \Qepsilon(t) + \Lambda(t)Q_{2\epsilon}(t)\Lambda(t)^{-1}.
  \end{align*}
 Consequently, using the results for $a = 1$ and  2, we see that the assertion holds.  In the same way we can prove the assertion for $a = 0,1,2,\dots$ by induction. 
 \par
 	Next we consider the case $a = -1,-2,\dots.$  From \eqref{4.6} we have 
 \begin{align*} 
0 =  &\Biggl[i\frac{\partial}{\partial t}- \widetilde{H}_{\epsilon}(t), \Lambda(t)^{-1}\Lambda(t)\Biggr] = \Biggl[i\frac{\partial}{\partial t}- \widetilde{H}_{\epsilon}(t), \Lambda(t)^{-1}\Biggr]\Lambda(t) \\
 & + \Lambda(t)^{-1}\Biggl[i\frac{\partial}{\partial t}- \widetilde{H}_{\epsilon}(t), \Lambda(t)\Biggr],
 \end{align*}
 which shows 
 \begin{align*} 
 &\Biggl[i\frac{\partial}{\partial t}- \widetilde{H}_{\epsilon}(t), \Lambda(t)^{-1}\Biggr]\Lambda(t) = - \Lambda(t)^{-1}\Biggl[i\frac{\partial}{\partial t}- \widetilde{H}_{\epsilon}(t), \Lambda(t)\Biggr]\Lambda(t)^{-1}\Lambda(t) \\
 &  =  - \Lambda(t)^{-1}\Qepsilon(t)\Lambda(t).
 \end{align*}
 Hence the assertion for $a = -1$ holds.  We consider the case $a = -2$.  From \eqref{4.6} we have 
 \begin{align*} 
 &\Biggl[i\frac{\partial}{\partial t}- \widetilde{H}_{\epsilon}(t), \Lambda(t)^{-2}\Biggr]\Lambda(t)^{2} = \Biggl[i\frac{\partial}{\partial t}- \widetilde{H}_{\epsilon}(t), \Lambda(t)^{-1}\Biggr]\Lambda(t) \\
 & + \Lambda(t)^{-1}\Biggl\{\Biggl[i\frac{\partial}{\partial t}- \widetilde{H}_{\epsilon}(t), \Lambda(t)^{-1}\Biggr]\Lambda(t)\Biggr\} \Lambda(t) \\
  & = Q_{-1\epsilon}(t) + \Lambda(t)^{-1}Q_{-1\epsilon}(t)\Lambda(t),
 \end{align*}
 which shows the assertion.  In the same way we can prove the assertion for $a = -1,-2,\dots$ by induction.  Thus, our proof is complete.
 \end{proof}
 	We consider the equation
 \begin{equation} \label{4.7}
i\frac{\partial u}{\partial t}(t) = \widetilde{H}_{\epsilon}(t)u(t) + f(t).
   \end{equation}
   \begin{pro}
   Suppose Assumption 2.2.
   Let $u_0 \in B^a_M \ (\adots)$ and $f(t) \in \Eta B^a_M)$.  Then, there exist solutions $u_{\epsilon}(t) \in \mathcal{E}_t^1([0,T];B^a_M)\ (\zeo)$ with $u_{\epsilon}(0) = u_0$ to \eqref{4.7} satisfying
 \begin{equation} \label{4.8}
\sup_{\zeo}\Vert u_{\epsilon}(t)\Vert_{a,M} \leq C_a\left(\Vert u_{0}\Vert_{a,M} + \int_0^t \Vert f(\theta)\Vert_{a,M} d\theta\right).
   \end{equation}
   \end{pro}
   \begin{proof}
   Applying Theorem 2.5 in Chapter 2 of \cite{Kumano-go} to \eqref{4.1}, we see that each of $\widetilde{H}_{\epsilon}(t)\ (\zeo)$ is written as the pseudo-differential operator with a symbol $p_{\epsilon}(t,x,\xi)$ satisfying
 \begin{equation*} 
\sup_{t,x,\xi}\left|p^{(\alpha)}_{\epsilon\ (\beta)} \txxi\right| \leq C_{\alpha\beta} < \infty
  \end{equation*}
  for all $\alpha$ and $\beta$, where $C_{\alpha\beta}$ may depend on $\zeo$.  Consequently, it follows from   Lemma 2.5  of \cite{Ichinose 1995} with $s = a, a = 2(M+1)$ and $b = 2$  that we have
  \begin{equation*}
  \sup_{0 \leq t \leq T}\Vert \widetilde{H}_{\epsilon}(t)f\Vert_{a,M} \leq C_{a\epsilon}\Vert f\Vert_{a,M}
  \end{equation*}
  for $\adots$ with constants $C_{a\epsilon} \geq 0$ dependent on $\zeo$.  Hence, noting that
   the equation \eqref{4.7} is equivalent to 
 \begin{equation*} 
iu(t) = iu_0 + \int_0^t\bigl\{\Hepsilon(\theta)u(\theta) + f(\theta)\bigr\}d\theta,
  \end{equation*}
  we can find a solution $\uepsilon(t) \in \mathcal{E}_t^1([0,T];B^a_M)$ by the successive iteration for each $\zeo$.  From \eqref{4.5} and \eqref{4.7} we have
 \begin{equation} \label{4.9}
i\frac{\partial }{\partial t}\Lambda(t)^a\uepsilon(t) = \widetilde{H}_{\epsilon}(t)\Lambda(t)^a\uepsilon(t) + Q_{a\epsilon}(t)\Lambda(t)^a\uepsilon(t) + \Lambda(t)^af(t).
   \end{equation}
   Applying the Calder\'on-Vaillancourt theorem (cf. p.224 in \cite{Kumano-go}),  from \eqref{3.25}, Propositions 3.3 and 3.5 we have
   $\Lambda(t)^a\uepsilon(t) \in \mathcal{E}_t^1([0,T];L^2)$ because of $\uepsilon(t) \in \mathcal{E}_t^1([0,T];B^a_M)$.  Noting  \eqref{4.1} and that $H(t)$ is symmetric on $L^2$, from \eqref{4.9} we have
 \begin{align*} 
& \frac{d}{dt}\Vert \Lambda(t)^a\uepsilon(t)\Vert^2 = 2 \rittaire \left(\frac{\partial}{\partial t}\Lambda(t)^a\uepsilon(t),\Lambda(t)^a\uepsilon(t)\right) \\
& = -2\bigl(K(t)X_{\epsilon}(t)\Lambda(t)^au_{\epsilon}(t), X_{\epsilon}(t)\Lambda(t)^au_{\epsilon}(t)\bigr)
\\
& \quad  -2\rittaire \bigl(iQ_{a\epsilon}(t)\Lambda(t)^a\uepsilon(t),\Lambda(t)^a\uepsilon(t)\bigr) - 2\rittaire \bigl(i\Lambda(t)^af(t),\Lambda(t)^a\uepsilon(t)\bigr).
  \end{align*}
Hence, using \eqref{1.3.1}, Proposition 4.2 and the Calder\'on-Vaillancourt theorem, we have
 \begin{align} \label{4.10}
 \frac{d}{dt}\Vert \Lambda(t)^a\uepsilon(t)\Vert^2 \leq 2C_a\left( \Vert \Lambda(t)^a\uepsilon(t)\Vert^2 + \Vert \Lambda(t)^af(t)\Vert \cdot\Vert \Lambda(t)^a\uepsilon(t)\Vert\right) 
  \end{align}
  with a constant $C_a$ independent of $\zeo$.   
  \par
  	For a moment take a constant $\eta > 0$ and set $v(t) := \bigl(\Vert\Lambda(t)^a\uepsilon(t)\Vert^2 + \eta\bigr)^{1/2}$, which is a positive, continuously differentiable  function with respect to $t$.  From \eqref{4.10} we have
 \begin{align*} 
 \frac{d}{dt}v(t)^2 \leq 2C_a\bigl( v(t)^2 + \Vert \Lambda(t)^af(t)\Vert v(t)\bigr) 
  \end{align*}
  and so $v'(t) \leq C_a\bigl( v(t) + \Vert \Lambda(t)^af(t)\Vert \bigr) $.  Hence we see
 \begin{align*} 
 v(t) \leq e^{C_at}v(0) + C_a\int_0^t e^{C_a(t-\theta)}\Vert\Lambda(\theta)^af(\theta)\Vert d\theta.
  \end{align*}
  Letting $\eta$ to $0$, we get
 \begin{align} \label{4.11}
 \Vert \Lambda(t)^a\uepsilon(t)\Vert \leq e^{C_at} \Vert \Lambda(0)^au_0\Vert + C_a\int_0^t e^{C_a(t-\theta)}\Vert\Lambda(\theta)^af(\theta)\Vert d\theta.  \end{align}
  Therefore, noting \eqref{3.25}, Propositions 3.3 and 3.5, we can prove \eqref{4.8} with another constant $C_a \geq 0$.
  \par
   \end{proof}
   The following has been proved in Lemma 3.1 of \cite{Ichinose 2012}.
   \begin{lem}
   Let $\adots$.  Then
   the embedding map from $B^{a+1}_M$ into $B^{a}_M$ is compact.
   \end{lem}
	Now, we will prove Theorem 2.1 under Assumption 2.2.  Our proof is similar to that of Theorem in \cite{Ichinose 1995}. 
	\par
	{\bf 1st step.}  Throughout 1st step we suppose $u_0 \in B^{a+1}_M$ and $f(t) \in \mathcal{E}_t^0([0,T]; \\
	B^{a+1}_M)$.  Let $\uepsilon(t) \in \mathcal{E}_t^1([0,T];B^{a+1}_M)\ (\zeo)$ be the solutions to \eqref{4.7} with $u(0) = u_0$, found in Proposition 4.3.  We see from \eqref{3.25}, \eqref{4.4}, Propositions 3.5 and 4.3 that the family $\bigl\{\uepsilon(t)\bigr\}_{\zeo}$
	is bounded in  $\mathcal{E}_t^0([0,T];B^{a+1}_M)$ and 
equi-continuous in  $\mathcal{E}_t^0([0,T];B^{a}_M)$ because 
 \begin{equation*} 
i\bigl\{\uepsilon(t)  - \uepsilon(t')\bigr\} = \int_{t'}^t\Hepsilon(\theta)\uepsilon(\theta)d\theta + \int_{t'}^tf(\theta)d\theta
  \end{equation*}
  and 
 \begin{equation*} 
\sup_{\zeo}\max_{0 \leq t \leq T}\Vert \widetilde{H}_{\epsilon}(t)u_{\epsilon}(t)\Vert_{a,M} \leq C_a \sup_{\zeo}\max_{0 \leq t \leq T}\Vert u_{\epsilon}(t)\Vert_{a+1,M} \leq C'_a \Vert u_0\Vert _{a+1,M}.
  \end{equation*}
Consequently, it follows from Lemma 4.4 that we can apply the Ascoli-Arzel\`{a} theorem to $\left\{\uepsilon(t)\right\}_{\zeo}$ in $\Eta B^a_M)$.  Then, there exist a sequence $\bigl\{\epsilon_j\bigr\}_{j=1}^{\infty}$ tending to zero and a function $u(t) \in \Eta B^a_M)$  such that 
 \begin{align} \label{4.12}
 \lim_{j \rightarrow \infty} u_{\epsilon_j}(t)  = u(t) \  \text{in}\ \Eta B^a_M).
 \end{align}
  Since from \eqref{4.7} we have
 \begin{align*} 
& u_{\epsilon_j}(t) = u_0 -i \int_0^t\widetilde{H}_{\epsilon_j}(\theta)u_{\epsilon_j}(\theta)d\theta - i \int_0^tf(\theta)d\theta \\
& = u_0 -i \int_0^t\widetilde{H}_{\epsilon_j}(\theta)u(\theta)d\theta 
-i \int_0^t\widetilde{H}_{\epsilon_j}(\theta)\bigl\{u_{\epsilon_j}(\theta)- u(\theta)\bigr\}d\theta - i \int_0^tf(\theta)d\theta,
  \end{align*}
  as in the proof of \eqref{3.14} in \cite{Ichinose 1995} from \eqref{3.25} and \eqref{4.4} we have
 \begin{equation*} 
u(t) = u_0 -i \int_0^t\widetilde{H}(\theta)u(\theta)d\theta - i \int_0^tf(\theta)d\theta
  \end{equation*}
  in $\mathcal{E}_t^0([0,T];B^{a-1}_M)$ by using Lemma 2.2 in \cite{Ichinose 1995}.  Therefore we see that $u(t)$ belongs to $\Eta B^a_M) \cap \mathcal{E}_t^1([0,T];B^{a-1}_M)$ and satisfies 
 \begin{equation} \label{4.13}
i\frac{\partial u}{\partial t}(t) = \widetilde{H}(t)u(t) + f(t)
   \end{equation}
with $u(0) = u_0$.  From \eqref{4.8} and \eqref{4.12} we also have
 \begin{equation} \label{4.14}
\Vert u(t)\Vert_{a,M} \leq C_a\left(\Vert u_{0}\Vert_{a,M} + \int_0^t \Vert f(\theta)\Vert_{a,M} d\theta\right).
   \end{equation}
   \par
   {\bf 2nd step.}  In this step we will prove the uniqueness of solutions to \eqref{1.1.1}  in $\Eta B^a_M) \cap \mathcal{E}_t^1([0,T];B^{a-1}_M)$ for any $\adots$. 
   \par
   Let $u(t) \in \Eta B^a_M) \cap \mathcal{E}_t^1([0,T];B^{a-1}_M)$ be a solution to \eqref{1.1.1}, i.e.
 \begin{equation*} 
i\frac{\partial u}{\partial t}(t) = \widetilde{H}(t)u(t)
   \end{equation*}
   with $u(0) = 0$.  We may assume $a \leq 0$ because of $B^{a+1}_M \subset B^a_M$.  Let 
   $g(t) \in \mathcal{E}_t^0([0,T];B^{-a+2}_M)$ be an arbitrary function and consider the backward Cauchy problem
 \begin{equation*} 
i\frac{\partial v}{\partial t}(t) = \big\{H(t) + iK(t)\bigr\}v(t) + g(t)
   \end{equation*}
with $v(T) = 0$.  Since \eqref{1.3.1} is assumed, as in the proof of the 1st step we can get a solution $v(t) \in \mathcal{E}_t^0([0,T];B^{-a+1}_M) \cap \mathcal{E}_t^1([0,T];B^{-a}_M)$.  Then we have
\begin{align*}
0 &= \int_0^T\left(i\frac{\partial u}{\partial t}(t) - \widetilde{H}(t)u(t), v(t)\right)dt \\
& =  \int_0^T\left(u(t), i\frac{\partial v}{\partial t}(t) -\big\{H(t) + iK(t)\bigr\}v(t)\right)dt = \int_0^T \bigl(u(t), g(t)\bigr)dt,
\end{align*}
which shows $u(t) = 0$.
\par
{\bf 3rd step.}  Let $u_0 \in B^a_M.$  We take $\bigl\{u_{0j}\bigr\}_{j=1}^{\infty}$ in $B^{a+1}_M$ such that $\lim_{j\to \infty}u_{0j} = u_0$ in $B^{a}_M$.  Let $u_j(t) \in \mathcal{E}_t^0([0,T];B^{a}_M) \cap \mathcal{E}_t^1([0,T];B^{a-1}_M)$ be the solution to \eqref{1.1.1} with $u(0) = u_{0j}$, uniquely determined in the above 2 steps.  Since $u_j(t) - u_k(t) \in \mathcal{E}_t^0([0,T];B^{a}_M) \cap \mathcal{E}_t^1([0,T];B^{a-1}_M)$ is the solution to \eqref{1.1.1} with $u(0) = u_{0j} - u_{0k} \in B^{a+1}_M$, from \eqref{4.14} we have
 \begin{equation} \label{4.15}
\Vert u_j(t) - u_k(t)\Vert_{a,M} \leq C_a\Vert u_{0j} - u_{0k}\Vert_{a,M}.
   \end{equation}
   Consequently, there exists a $u(t) \in \mathcal{E}_t^0([0,T];B^{a}_M)$ such that $\lim_{j\to \infty}u_j(t) = u(t)$ in $\mathcal{E}_t^0([0,T];B^{a}_M)$.  Since
 \begin{equation*} 
u_{j}(t) = u_{0j} -i \int_0^t\widetilde{H}(\theta)u_{j}(\theta)d\theta,
  \end{equation*}
   $u(t)$ belongs to $\mathcal{E}_t^0([0,T];B^{a}_M) \cap \mathcal{E}_t^1([0,T];B^{a-1}_M)$ and satisfies \eqref{1.1.1} with $u(0) = u_0$.  We can also prove \eqref{2.8.1} because $\Vert u_j(t)\Vert_{a,M} \leq C_a\Vert u_{0j}\Vert_{a,M}$ holds from \eqref{4.14}.  Thus we have completed the proof of Theorem 2.1 under Assumption 2.2.
  \par
 We will prove Theorem 2.1 under Assumption 2.1.  We define $\Lambda(X,D_x)$ by \eqref{3.14} where $h_s(x,\xi)$
 is replaced with $|x|^2 + |\xi|^2$.  We also define $\chi_{\epsilon}(x,\xi)$ by \eqref{3.15} where $h(x,\xi)$
 is replaced with $|x|^2 + |\xi|^2$.  Then it is easy to show the same assertions as in Lemma 3.4.  Noting Proposition 3.5, we can also prove the same assertions as in Lemma 4.1 and Propositions 4.2 - 4.3 where $M = 0$.  Consequently, we can prove Theorem 2.1 under Assumption 2.1 as in the proof of that under Assumption 2.2.
	Thus, our proof of Theorem 2.1 is complete. 
	\vspace{0.3cm}
\par
	Next, we will prove (2) of Theorem 2.2.  The proof of (1) of Theorem 2.2 can be given in the same way.  Our proof below is similar to that of Theorem 4.1 in \cite{Ichinose 2012}.  For simplicity we write $\Vert f \Vert_{a,M}$ as $\Vert f \Vert_{a}$ hereafter in this section.
	\par
	Let $\utrho\ (\rho \in \mathcal{O})$ be the solutions to \eqref{1.1.1} with $u(0;\rho) = u_0 \in B^a_M \ (\adots)$.  Then, following the proof of Theorem 2.1, under the assumptions of Theorem 2.2 we have
 \begin{equation} \label{4.16}
\sup_{\rho \in \mathcal{O}}\Vert \utrho\Vert_a \leq C_a\Vert u_0\Vert_a\quad (0 \leq t \leq T).
   \end{equation}
   \par
   	We first assume $u_0 \in B^{a+1}_M$.  Then from \eqref{4.16} we have
 \begin{equation*} 
\sup_{\rho \in \mathcal{O}}\Vert \utrho\Vert_{a+1} \leq C_{a+1}\Vert u_0\Vert_{a+1}
   \end{equation*}
    and hence as in the 1st step of the proof of Theorem 2.1
 \begin{equation*} 
\Vert \utrho - u(t';\rho)\Vert_{a} \leq C'_{a}|t-t'|\Vert u_0\Vert_{a+1}
   \end{equation*}
   with a constant $C'_a$ independent of $\rho$.  Consequently, we see that the family $\big\{\utrho\bigr\}_{\rho \in \mathcal{O}}$ is bounded in $\mathcal{E}_t^0([0,T];B^{a+1}_M)$ and equi-continuous in $\mathcal{E}_t^0([0,T];B^{a}_M)$.  Let $\rho_j \to \rho$
   in $\mathcal{O}$ as $j \to \infty$.  Noting Lemma 4.4, we can apply the Ascoli-Arzel\`{a} theorem to $\big\{u(t;\rho_j)\bigr\}_{j=1}^{\infty}$ in $\mathcal{E}_t^0([0,T];B^{a}_M)$.  Then, there exist a subsequence $\big\{j_k\bigr\}_{k=1}^{\infty}$ and a function $v(t) \in \Eta B^{a}_M)$ such that $\lim_{k\to \infty}u(t;\rho_{j_k}) = v(t) $ in $\Eta B^{a}_M)$.  As in the proof of \eqref{4.13} we see that $v(t)$ belongs to $\mathcal{E}_t^1([0,T];B^{a-1}_M)$ and satisfies \eqref{1.1.1} with $u(0) = u_0$.  The uniqueness of solutions to \eqref{1.1.1} gives $v(t) = \utrho$, which shows $\lim_{k\to \infty}u(t;\rho_{j_k}) = \utrho$.  Using the uniqueness of solutions to \eqref{1.1.1} again,  we have
 \begin{equation*} 
\lim_{j \to \infty}u(t;\rho_{j}) = \utrho\ \text{in}\ \Eta B^{a}_M).
   \end{equation*}
   Therefore we see that the mapping $: \mathcal{O} \ni \rho \to \utrho \in \Eta B^{a}_M)$ is continuous.
   \par
   Now let  $u_0 \in B^a_M$ and $\utrho\ (\rho \in \mathcal{O})$ the solutions to \eqref{1.1.1} with $u(0) = u_0$.   We take $\bigl\{u_{0k}\bigr\}_{k=1}^{\infty}$ in $B^{a+1}_M$ such that $\lim_{k\to \infty}u_{0k} = u_0$ in $B^{a}_M$ and let $u_k(t;\rho) \in \mathcal{E}_t^0([0,T];B^{a}_M) \cap \mathcal{E}_t^1([0,T];B^{a-1}_M)$ be the solutions to \eqref{1.1.1} with $u(0) = u_{0k}$. Then, from \eqref{4.16} we have
 \begin{equation*} 
\sup_{\rho}\max_t \Vert u_k(t;\rho) - u(t;\rho)\Vert_a \leq C_a\Vert u_{0k} - u_{0}\Vert_a,
   \end{equation*}
which shows that $\utrho$ is continuous in $\Eta B^{a}_M)$ with respect to $\rho \in \mathcal{O}$ because so is $ u_k(t;\rho)$.  Thus our proof of Theorem 2.2 is complete. \vspace{0.3cm} 
\par
	In the end of this section we will prove (2) of Theorem 2.3.  The proof of (1) is given in the same way. Our proof below is similar to that of Theorem 2.3 in \cite{Ichinose 2012}.
	\par
	Let $u_0 \in B^a_M\ (\adots)$ and $f(t) \in \Eta B^{a}_M)$.  Then, we see that  there exists the unique solution $u(t) \in \Eta B^{a}_M) \cap \mathcal{E}_t^1([0,T];B^{a-1}_M)$ to \eqref{4.13} with $u(0) = u_0$, which satisfies 
 \begin{equation} \label{4.17}
\Vert u(t)\Vert_a \leq C_a\left(\Vert u_{0}\Vert_a + \int_0^t \Vert f(\theta)\Vert_a d\theta\right).
   \end{equation}
   Its proof can be completed by using \eqref{4.14} as in the 3rd step of the proof of Theorem 2.1.
	\par
	Let $u_0 \in B^{a+1}_M$ and $\utrho \in \mathcal{E}_t^0([0,T];B^{a+1}_M) \cap \mathcal{E}_t^1([0,T];B^{a}_M)\ (\rho \in \mathcal{O})$ the solutions to \eqref{1.1.1} with $u(0) = u_0$.  Let $\rho \in \mathcal{O}$ be fixed and $\tau \not= 0$ small constants such that $\rho + \tau \in \mathcal{O}.$  We set
 \begin{equation} \label{4.18}
w_{\tau} (t;\rho) := \frac{u(t;\rho+\tau) - u(t;\rho)}{\tau},
 \end{equation}
   which belong to $\Eta B^{a}_M) \cap \mathcal{E}_t^1([0,T];B^{a-1}_M)$.
   Then we have $w_{\tau}(0;\rho) = 0$ and from \eqref{1.1.1}
\begin{equation} \label{4.19}
 i\frac{\partial}{\partial t}w_{\tau} (t;\rho) = \widetilde{H}(t;\rho)w_{\tau} (t;\rho) + \int_0^1 \frac{\partial  \widetilde{H}}{\partial \rho}(t;\rho+\theta\tau)d\theta\, u(t;\rho+\tau).
 \end{equation}
 Hence, noting \eqref{2.15.1} - \eqref{2.17.1}, from \eqref{4.16} and \eqref{4.17} we get
\begin{align*}
& \Vert w_{\tau} (t;\rho)\Vert_a \leq C_a\int_0^t\int_0^1 \left\Vert \frac{\partial  \widetilde{H}}{\partial\rho}(t';\rho+\theta\tau)u(t';\rho+\tau)\right\Vert_ad\theta dt' \\
& \leq C'_a\int_0^t \left\Vert u(t';\rho+\tau)\right\Vert_{a+1} dt' \leq C''_a \left\Vert u_0\right\Vert_{a+1}.
\end{align*}
Consequently, 
 \begin{equation} \label{4.20}
\sup_{\tau}\Vert w_{\tau} (t;\rho)\Vert_a \leq C_a \left\Vert u_0\right\Vert_{a+1}
\end{equation}
   with another constant $C_a$ independent of $\rho \in \mathcal{O}$.
   \par
     We first assume $u_0 \in B^{a+2}_M$.  From  \eqref{4.20} we have
 \begin{equation*} 
\sup_{\tau}\Vert w_{\tau} (t;\rho)\Vert_{a+1} \leq C_{a+1} \left\Vert u_0\right\Vert_{a+2}.
\end{equation*}
  Thereby from \eqref{4.16} and \eqref{4.19} we have  
   \begin{equation*} 
\sup_{\tau}\Vert w_{\tau} (t;\rho) - w_{\tau} (t';\rho)\Vert_{a} \leq C'_a |t - t'|\left \Vert u_0\right\Vert_{a+2}
\end{equation*}
as in the 1st step of the proof of Theorem 2.1. Hence we can apply the Ascoli-Arzel\`{a} theorem to $\bigl\{w_{\tau}(t;\rho)\bigr\}_{\tau}$ in $\Eta B^a_M)$.  In addition, we can use  the uniqueness of solutions to \eqref{2.14.1} or \eqref{4.13}.  Then, using Theorem 2.2, 
as in the 3rd step of the proof of Theorem 2.1 we can prove  from \eqref{4.19} that there exists a function $\wtrho \in \Eta B^a_M)\cap \mathcal{E}_t^1([0,T];B^{a-1}_M)$ satisfying \eqref{2.14.1} with $w(0) = 0$ and 
 \begin{equation} \label{4.21}
\lim_{\tau\to 0}w_{\tau} (t;\rho) = \wtrho \ \text{in}\ \Eta B^a_M).
\end{equation}
\par
	Now let  $u_0 \in B^{a+1}_M$.  Let $\utrho$ be the solution to \eqref{1.1.1} with $u(0) = u_0$ and define $w_{\tau}(t;\rho)$  by \eqref{4.18}.  We take $\bigl\{u_{0k}\bigr\}_{k=1}^{\infty}$ in $B^{a+2}_M$ such that $\lim_{k\to \infty}u_{0k} = u_0$ in $B^{a+1}_M$.  Let $u_k(t;\rho) \in \mathcal{E}_t^0([0,T];B^{a+1}_M) \cap \mathcal{E}_t^1([0,T];B^{a}_M)$ be the solution to \eqref{1.1.1} with $u(0) = u_{0k}$.  We define $w_{k\tau}$ by \eqref{4.18} with $u = u_k$ and $w_{k}$  by \eqref{4.21} with $w_{\tau} = w_{k\tau}$.  From \eqref{4.19} we have
	\begin{align*}
 & i\frac{\partial}{\partial t}\bigl\{w_{k\tau} (t;\rho) - w_{\tau} (t;\rho)\bigr\}= 
\widetilde{H}(t;\rho)\bigl\{w_{k\tau} (t;\rho) - w_{\tau} (t;\rho)\bigr\} \\
 &\quad  + \int_0^1 \frac{\partial\widetilde{H}}{\partial \rho}(t;\rho+\theta\tau)d\theta\, \bigl\{u_k(t;\rho+\tau) - u (t;\rho+\tau)\bigr\}
 \end{align*}
 and $w_{k\tau} - w_{\tau} \in \Eta B^a_M)\cap \mathcal{E}_t^1([0,T];B^{a-1}_M) $.  Hence, as in the proof of \eqref{4.20}  we have
 \begin{equation} \label{4.22}
\sup_{\tau}\Vert w_{k\tau}(t;\rho) - w_{\tau}(t;\rho)\Vert_a \leq C_a\Vert u_{0k} - u_{0}\Vert_{a+1}
   \end{equation}
with the constant $C_a$ in \eqref{4.20}.
As noted in the first part of this proof, there exists
 the solution $\wtrho \in \mathcal{E}_t^0([0,T];B^{a}_M) \cap \mathcal{E}_t^1([0,T];B^{a-1}_M)$  to \eqref{2.14.1} with $w(0) = 0$ because of $\partial_{\rho}\widetilde{H}(t;\rho)u(t;\rho) 
  \in \Eta B^a_M)$.  Since both of $w_k$ and $w$ are the solutions to \eqref{2.14.1}, as in the proof of \eqref{4.22} we  have
 \begin{equation} \label{4.23}
 \Vert w_{k}(t;\rho) - w(t;\rho)\Vert_a \leq C_a\Vert u_{0k} - u_{0}\Vert_{a+1}.
   \end{equation}
Consequently, we have
\begin{align*}
& \Vert w_{\tau} (t;\rho)- w(t;\rho)\Vert_a \leq \Vert w_{\tau} - w_{k\tau}\Vert_a + \Vert w_{k\tau}- w_k\Vert_a + \Vert w_{k} - w\Vert_a \\
& \leq 2C_a \Vert u_{0k} - u_{0}\Vert_{a+1} + \Vert w_{k\tau}- w_k\Vert_a.
\end{align*}
Hence we see from \eqref{4.21} that  we get  $\overline{\lim}_{\tau \to 0}\max_{\,t}\Vert w_{\tau}- w\Vert_a \leq 2C_a\Vert u_{0k} - u_{0}\Vert_{a+1}$, which shows
 \begin{equation} \label{4.24}
\lim_{\tau \to 0}\max_{0 \leq t \leq T}\Vert w_{\tau}(t;\rho) - w(t;\rho)\Vert_a = 0.
   \end{equation}
We also have \eqref{2.13.1} from \eqref{4.20} and \eqref{4.24}.
\par
	In the end of this proof we will prove that $\wtrho = \partial_{\rho}\utrho$ for $u_0 \in B^{a+1}_M$ is continuous in $\mathcal{E}_t^0([0,T];B^{a}_M)$ with respect to $\rho \in \mathcal{O}$.  
	We first assume $u_0 \in B^{a+2}_M$.  Then we have \eqref{4.16} where $a$ is replaced with $a+2$.  Since $\wtrho$ is the solution to \eqref{2.14.1} with $w(0) = 0$, we see from \eqref{4.17} as in the proof of \eqref{4.21} that the family $\bigl\{\wtrho\bigr\}_{\rho \in \mathcal{O}}$ is bounded in $\mathcal{E}_t^0([0,T];B^{a+1}_M)$ and equi-continuous in $\Eta B^a_M)$. 
Hence, noting that $\utrho$ is continuous in $\mathcal{E}_t^0([0,T];B^{a+2}_M)$  with respect to $\rho$, we see that so is $\wtrho$ in $\Eta B^a_M)$ as in the proof of Theorem 2.2.  Now let $u_0 \in B^{a+1}_M$.  We take $\bigl\{u_{0k}\bigr\}_{k=1}^{\infty}$ in $B^{a+2}_M$ such that $\lim_{k\to \infty}u_{0k} = u_0$ in $B^{a+1}_M$ and write  as $w_k(t,\rho)$ the solutions to \eqref{2.14.1} with $u(t;\rho) = u_k(t;\rho)$ and $w(0) = 0$.  Then we have \eqref{4.23}, which shows that $\wtrho$ is continuous with respect to $\rho \in \mathcal{O}$ in $\Eta B^a_M)$ because so is $w_k(t;\rho)$.  Therefore, our proof of Theorem 2.3 is complete.
%
   \section{Proofs of Theorems 2.4 - 2.6}
	In this section we will study the 4-particle systems \eqref{1.7.1}.   Let $(x,\xi) \in \bR^{2d}$ and write 
\begin{equation} \label{5.1}
 h_k(t,x,\xi) :=  \frac{1}{2m_k}|\xi - A^{(k)}(t,x)|^2 + V_k(t,x) \ (k = 1,2,3,4)
 \end{equation}
  and 
\begin{equation} \label{5.2}
 l_k(x,\xi) :=  \frac{1}{2m_k}|\xi |^2 + <x>^2 \ (k = 3,4).
 \end{equation}
  We set 
\begin{equation} \label{5.3}
 \widehat{h}(t,z,\zeta) :=  \sum_{k=1}^2 h_k(t,x^{(k)},\xi^{(k)}) + W_{12}(t,x^{(1)} -x^{(2)}) + \sum_{k=3}^4l_k(x^{(k)},\xi^{(k)})
 \end{equation}
 and write 
\begin{equation} \label{5.4}
 \widehat{H}(t) :=  \widehat{H}\left(t,\frac{Z+Z'}{2},D_z\right),
 \end{equation}
 where $z = (x^{(1)},x^{(2)},x^{(3)},x^{(4)})$ and $\zeta = (\xi^{(1)},\xi^{(2)},\xi^{(3)},\xi^{(4)})$ in $\bR^{4d}$.   We  also set 
\begin{equation} \label{5.5}
\widehat{h}_s(t,z,\zeta) :=   \widehat{h}(t,z,\zeta) + i\sum_{k=1}^2\frac{1}{2m_k}\nabla\cdot A^{(k)}(t,x^{(k)})
 \end{equation}
 and 
\begin{equation*} 
 p_{\mu}(t,z,\zeta) :=  \frac{1}{\mu + \widehat{h}_s(t,z,\zeta)}
 \end{equation*}
 for large $\mu$ as in \eqref{3.2} and \eqref{3.3}, respectively.
 \begin{lem}
 Assume \eqref{1.5.1}, \eqref{2.3.1} and \eqref{2.5.1} - \eqref{2.6.1} for $(V_k,A^{(k)})\ (k = 1,2)$ with $M = M_k$ and \eqref{2.18.1} - \eqref{2.19.1} for $W_{12}$.  Then, there exist a constant $\mu^* \geq 0$ and functions $r_{\mu}(t,z,\zeta)\ (\mu \geq \mu^*)$ such that
 \begin{equation} \label{5.5-2}
	\mu^* + \rittaire	\, \widehat{h}_s(t,z,\zeta) \geq C_0^*(<\zeta>^2 + \Phi(z)^2),
	\end{equation}
	\begin{equation} \label{5.6}
	\bigl[\mu + \widehat{H}(t)\bigr] P_{\mu}\tzdz = I + R_{\mu}\tzdz,
	\end{equation}
	\begin{equation} \label{5.7}
	\left|r^{(\alpha)}_{\mu\ (\beta)}(t,z,\zeta)\right| \leq  C_{\alpha\beta}\, \mu^{-1/2}
	\end{equation}
	in $[0,T]\times \bR^{8d}$ for all $\alpha, \beta$ and $\mu \geq \mu^*$ with constants $C_0^* > 0$ and $C_{\alpha\beta}$ independent of $\mu$, where
	\begin{equation} \label{5.8}
	\Phi(z) =  \sum_{k=1}^2<x^{(k)}>^{M_k+1} + \sum_{k=3}^4<x^{(k)}>.
	\end{equation}
 \end{lem}
 \begin{proof}
 As in the proof of  \eqref{3.6} we see
 \begin{equation*} 
  \rittaire \widehat{h}_s(t,z,\zeta) = \widehat{h}(t,z,\zeta) \geq  C_0\bigl(<\zeta>^2  + \Phi(z)^2\bigr) - |W_{12}(t,x^{(1)} -x^{(2)})| - C_1
   \end{equation*}
   with constants $C_0 > 0$ and $C_1 \geq 0$.  Hence, using the assumption \eqref{2.18.1}, we can determine constants $\mu^* \geq 0$ and $C^*_0 > 0$ satisfying \eqref{5.5-2}.  Then, using \eqref{5.5-2},
    as in the proof of \eqref{3.7} for $\mu \geq \mu^*$ we have
   \begin{align} \label{5.10}
  & r_{\mu}(t,z,\zeta)   =  \sum_{|\alpha| = 1} \int_0^1 d\theta\  \text{Os}-\iint e^{-iy\cdot\eta}
 <y>^{-2l_0}  <D_{\eta}>^{2l_0}<\eta>^{-2l_1} <D_y>^{2l_1} \notag\\
 & \hspace{2cm} \cdot \widehat{h}_s^{(\alpha)}(t,z,\zeta+\theta\eta)  p_{\mu (\alpha)}(t,z+y,\xi)dy\hdbar\eta
 \end{align}
  for large integers $l_0$ and $l_1$.  In addition, as in the proofs of \eqref{3.8} - \eqref{3.9} we can show 
 \begin{equation} \label{5.11}
 |\widehat{h}^{(\alpha)}_{s\,(\beta)}(t,z+y,\zeta)| \leq  C_{\alpha\beta} \left(<\zeta>^2 + \Phi(z+y)^2\right) 
 \end{equation}
 for all $\alpha$ and $|\beta|\geq 1$, and 
 \begin{equation} \label{5.12}
 |\widehat{h}^{(\alpha)}_{s\, (\beta)}(t,z,\zeta + \theta\eta)| \leq  C_{\alpha\beta} \bigl(<\zeta> + \Phi(z)\bigr) <\eta>
 \end{equation}
 for $|\alpha|\geq 1$ and all $\beta$.  Therefore, we can complete the proof of Lemma 5.1 from \eqref{5.10} - \eqref{5.12} as in the proof of Lemma 3.2.
 \end{proof}
 	We can easily see from \eqref{5.11} and \eqref{5.12} as in the proof of \eqref{3.25} that under the assumptions of Lemma 5.1 we have
 \begin{equation} \label{5.12-2}
 |\widehat{h}^{(\alpha)}_{s\, (\beta)}(t,z,\zeta)| \leq  C_{\alpha\beta}\bigl(<\zeta>^2 + \Phi(z)^2\bigr) 
 \end{equation}
 for all $\alpha$ and  $\beta$.
 \begin{pro}
 Under the assumptions of Lemma 5.1 there exist a constant $\mu \geq \mu^*$ and a function $w(t,z,\zeta)$ satisfying 
 \begin{equation} \label{5.13}
  |w^{(\alpha)}_{\ \  (\beta)}(t,z,\zeta)| \leq C_{\alpha\beta}\left(<\zeta>^2 + \Phi(z)^2\right)^{-1}
  \end{equation}
  for all $\alpha, \beta$ and 
 \begin{equation} \label{5.14}
 W\tzdz = \bigl(\mu+ \widehat{H}(t)\bigr)^{-1}.
   \end{equation}
 \end{pro}
 \begin{proof}
If $\mu \geq \mu^*$, from \eqref{5.5-2} and  \eqref{5.11} - \eqref{5.12} we see
 \begin{equation*} 
 |p^{(\alpha)}_{\mu\, (\beta)}(t,z,\zeta)| \leq  C_{\alpha\beta}\bigl(<\zeta>^2 + \Phi(z)^2\bigr)^{-1}
 \end{equation*}
 for all $\alpha$ and  $\beta$ as in the proof of Proposition 3.3. Hence, using Lemma 5.1, we can prove Proposition 5.2 as in the proof of Proposition 3.3.
 \end{proof}
	We take a $\mu$ stated in Proposition 5.2 and fix it hereafter.  We set
 \begin{equation} \label{5.15}
 \lambda(t,z,\zeta)  := \mu + \widehat{h}_s(t,z,\zeta)
   \end{equation}
   as in \eqref{3.13}.  Then, from \eqref{5.1} - \eqref{5.5} we have 
 \begin{align} \label{5.16}
 \Lambda(t) & = \Lambda\tzdz = \mu + \widehat{H}_s\tzdz =  \mu + \widehat{H}(t) \notag\\
 & = \mu + H_1(t) + H_2(t) + W_{12}(t) + L_3(t) + L_4(t),
  \end{align}
   where $H_k(t)$ are the operators defined by \eqref{1.7.1} and $L_k(t)$ the pseudo-differential operators with the symbols $l_k(x^{(k)},\xi^{(k)})$ defined by \eqref{5.2}.  Using the real-valued function $\widehat{h}(t,z,\zeta)$ defined by \eqref{5.3}, we determine
 \begin{equation} \label{5.17}
 \chi_{\epsilon}(t,z,\zeta) = \chi\bigl(\epsilon(\mu + \widehat{h}(t,z,\zeta)\bigr)
   \end{equation}
   with constants $\zeo$ and $\chi \in \mathcal{S}(\bR)$ such that $\chi(0) = 1$ as in \eqref{3.15}.
   \par
   	Lemmas 5.3 and 5.4 below are crucial in this section.
	\begin{lem}
	Under the assumptions of Lemma 5.1 there exist functions  $\omega_{\epsilon}(t,z,\zeta)\ (\zeo)$ in $[0,T] \times \bR^{8d}$ satisfying
 \begin{equation} \label{5.18}
 \sup_{0 < \epsilon \leq 1}\sup_{t,z,\zeta} |\omega^{(\alpha)}_{\epsilon\  (\beta)}(t,z,\zeta)| \leq C_{\alpha\beta} < \infty
  \end{equation}
  for all $\alpha, \beta$ and 
 \begin{equation} \label{5.19}
\Omega_{\epsilon}\tzdz = \Bigl[X_{\epsilon}\tzdz, \Lambda\tzdz\Bigr].
   \end{equation}
	\end{lem}
	\begin{proof}
	As in the proof of \eqref{3.18} we see
   \begin{align} \label{5.20}
  &\omega_{\epsilon}(t,z,\zeta) = \sum_{|\alpha| = 1} \Bigl\{\chi_{\epsilon}^{(\alpha)}(t,z,\zeta)\lambda_{(\alpha)}(t,z,\zeta) 
  - \lambda^{(\alpha)}(t,z,\zeta)\chi_{\epsilon(\alpha)}(t,z,\zeta) \Bigr\}\notag\\
 &\quad   + 2 \sum_{|\gamma| = 2}\frac{1}{\gamma\, !}\int_0^1 (1 - \theta)d\theta\  \text{Os}-\iint e^{-iy\cdot\eta}
 \Bigl\{\chi_{\epsilon}^{(\gamma)}(t,z,\zeta+\theta\eta)\lambda_{(\gamma)}(t,z+y,\zeta)  \notag\\
 & \qquad \quad 
  - \lambda^{(\gamma)}(t,z,\zeta+\theta\eta)\chi_{\epsilon(\gamma)}(t,z+y,\xi) \Bigr\}dy\hdbar\eta 
  \equiv I_{1\epsilon} + I_{2\epsilon}.
 \end{align}
From \eqref{5.5}, \eqref{5.15} and \eqref{5.17} we can write
   \begin{align} \label{5.21}
  & I_{1\epsilon}(t,z,\zeta) = \epsilon\chi\,'(\epsilon(\mu +\widehat{h}))\sum_{|\alpha| = 1}\Bigl\{\widehat{h}^{(\alpha)}\widehat{h}_{s(\alpha)}
  - \widehat{h}^{(\alpha)}_{s}\widehat{h}_{(\alpha)} \Bigr\}\notag\\
 & = i \epsilon\chi\,'(\epsilon(\mu +\widehat{h}))\sum_{|\alpha| = 1}
 \widehat{h}^{(\alpha)}_s(t,z,\zeta)\sum_{k=1}^2 \frac{1}{2m_k}(-i\partial_z)^{\alpha}\,\nabla\cdot A^{(k)}\bigl(t,x^{(k)},\xi^{(k)}\bigr).
 \end{align}
   From \eqref{5.5-2} we have
 \begin{equation} \label{5.21-2}
 \bigl (\mu + \widehat{h}(t,z,\zeta)\bigr)^{-1} \leq  C_0\bigl(<\zeta>^2  + \Phi(z)^2\bigr)^{-1}
   \end{equation}
 because of $\widehat{h} = \rittaire \widehat{h}_s$.  Hence, together with \eqref{2.6.1} and \eqref{5.12} we can prove $\sup_{\epsilon}\sup_{t,z,\zeta} 
   |I_{1\epsilon}| < \infty$ as in the proof of \eqref{3.20}.  In the same way we can prove
 \begin{equation} \label{5.22}
 \sup_{0 < \epsilon \leq 1}\sup_{t,z,\zeta} |I^{(\alpha)}_{1\epsilon\  (\beta)}(t,z,\zeta)| \leq C_{\alpha\beta} < \infty
  \end{equation}
for all $\alpha$ and $\beta$.  \par
   Let $|\gamma| = 2$.  Then, from \eqref{5.5} and \eqref{5.11} - \eqref{5.12} we have the similar inequalities
 \begin{equation*} 
\sup_{\zeo}  |\chi^{(\alpha+\gamma)}_{\epsilon\  (\beta)}(t,z,\zeta)| \leq C_{\alpha\beta}\left(<\zeta>^2 + \Phi(z)^2\right)^{-1}
  \end{equation*}
  and 
 \begin{equation*} 
\sup_{\zeo}  |\chi^{(\alpha)}_{\epsilon\ (\beta+\gamma)}(t,z,\zeta)| \leq C_{\alpha\beta} < \infty
  \end{equation*}
  to \eqref{3.21} and \eqref{3.22} for all $\alpha$ and $\beta$, respectively.  Consequently, noting that $\lambda^{(\gamma)}(t,z,\zeta) = \widehat{h}_s^{(\gamma)}(t,z,\zeta)$ are constants, from \eqref{5.20} we can prove
 \begin{equation*} 
 \sup_{0 < \epsilon \leq 1}\sup_{t,z,\zeta} |I^{(\alpha)}_{2\epsilon\  (\beta)}(t,z,\zeta)| \leq C_{\alpha\beta} < \infty
  \end{equation*}
for all $\alpha$ and $\beta$ as in the proof of Lemma 3.4, which completes the proof together with \eqref{5.20} and \eqref{5.22}.
	\end{proof}
Let $\widetilde{H}(t)$ be the operator defined by \eqref{1.7.1}.
\begin{lem}
Besides the assumptions of Lemma 5.1 we suppose that each $(V_k,A^{(k)})\ (k = 3,4)$ satisfies \eqref{1.4.1}  and each 
$W_{ij}(t,x)\ (1 \leq i <j \leq 4)$ except $W_{12}$ satisfies \eqref{2.20.1}.  In addition, we suppose that $k_l(t,x,\xi)\ (l = 1,2)$ and $k_l(t,x,\xi)\ (l = 3,4)$ satisfy \eqref{2.2.1}  with $M = M_l$ and \eqref{2.1.1}, respectively.
Then,  there exists a function $\tilde{q}(t,z,\zeta)$  satisfying
 \begin{equation} \label{5.23}
 \sup_{t,z,\zeta} |\tilde{q}^{\,(\alpha)}_{\ \,(\beta)}(t,z,\zeta)| \leq C_{\alpha\beta} < \infty
  \end{equation}
  for all $\alpha, \beta$ and 
 \begin{equation} \label{5.24}
\widetilde{Q}\tzdz = \Bigl[\Lambda(t), \widetilde{H}(t)\Bigr]\Lambda(t)^{-1}.
   \end{equation}
\end{lem}
\begin{proof}
We write $\widetilde{H}(t)$   as
 \begin{equation} \label{5.25}
 \widetilde{H}(t) = \sum_{k=1}^4 \widetilde{H}_k(t) + W_{12}(t) + \sum \,'\,W_{ij}(t),
   \end{equation}
  where $ \widetilde{H}_k(t) = H_k(t) - iK_k(t).$
  Then from \eqref{5.16} we see
 \begin{align} \label{5.26}
&[\widetilde{H}(t), \Lambda(t)] = [( \widetilde{H}_1 +  \widetilde{H}_2 + W_{12}) +  \widetilde{H}_3 +  \widetilde{H}_4 + \sum\, '\,W_{ij}, (H_1 + H_2 + W_{12})\notag\\
&  + L_3 + L_4] = \Bigl(-i[K_1,H_1] -i[K_2,H_2] -i[K_1+ K_2,W_{12}] + [\widetilde{H}_3,L_3] + [ \widetilde{H}_4,L_4]\Bigr)
\notag \\ &
 + 
 \left[\sum\, '\,W_{ij}, H_1 + H_2 + L_3 + L_4\right] \equiv I_1(t) + I_2(t).
   \end{align}
 As in the proof of Lemma 4.1, we can prove that  $I_1(t)\Lambda(t)^{-1}$ is written as the pseudo-differential operator with a symbol satisfying \eqref{5.23}.
  \par
  We can easily see that $m_1\bigl[W_{13}(t),H_1(t)\bigr]$ is written as the pseudo-differential operator with the symbol
 \begin{align} \label{5.27}
& \tilde{q}_1(t,z,\zeta)  = i\frac{\partial W_{13}}{\partial x} (t,x^{(1)} - x^{(3)})\cdot \xi^{(1)} + \frac{1}{2}\Delta_xW_{13} (t,x^{(1)} - x^{(3)}) \notag \\
& - i A^{(1)}(t,x^{(1)})\cdot \frac{\partial W_{13}}{\partial x} (t,x^{(1)} - x^{(3)}).
  \end{align}
  Hence from the assumptions we have
 \begin{align*} 
& |\tilde{q}_1(t,z,\zeta)| \leq C_1\bigl(<\xi^{(1)}>^2 + <x^{(1)} - x^{(3)}>^2 + <x^{(1)}>^{M_1+1}<x^{(1)} - x^{(3)}>\bigr)
\notag \\
& \leq C_2\bigl(<\xi^{(1)}>^2 + <x^{(1)}>^{M_1+2} + <x^{(1)}>^{2(M_1+1)}+ < x^{(3)}>^2\bigr) \\
& \leq C_3\bigl(<\xi^{(1)}>^2 + <x^{(1)}>^{2(M_1+1)}+ < x^{(3)}>^2\bigr).
  \end{align*}
  In the same way we have
 \begin{equation} \label{5.28}
 |\tilde{q}^{\,(\alpha)}_{1 \,(\beta)}(t,z,\zeta)| \leq C_{\alpha\beta} \bigl(<\zeta>^2 + \Phi(z)^2 \bigr)
  \end{equation}
  for all $\alpha$ and $ \beta$.  Consequently, by Proposition 5.2 we see that $\bigl[W_{13}(t),H_1(t)\bigr]\Lambda(t)^{-1}$ is written as the pseudo-differential operator with a symbol satisfying \eqref{5.23}.  In the same way  we can complete the proof of Proposition 5.4.
\end{proof}
	Using the function $\chi_{\epsilon}(t,z,\zeta)$ defined by \eqref{5.17},  we define 
 \begin{equation} \label{5.29}
\widetilde{H}_{\epsilon}(t) := X_{\epsilon}\tzdz^{\dagger}\widetilde{H}(t)X_{\epsilon}\tzdz
   \end{equation}
as in \eqref{4.1}.
\begin{lem}
Under Assumption 2.3
 there exist functions $\qepsilon(t,z,\zeta)\ (\zeo)$ satisfying \eqref{5.18} and 
 \begin{align} \label{5.30}
Q_{\epsilon}\tzdz = &\Bigl[\Lambda\tzdz, \widetilde{H}_{\epsilon}(t)\Bigr]\Lambda\tzdz^{-1} \notag \\
&  + i\frac{\partial\Lambda}{\partial t}\tzdz\Lambda\tzdz^{-1}.
   \end{align}
\end{lem}
\begin{proof}
From \eqref{5.29} and $\Lambda(t)^{\dagger} = \Lambda(t)$ we have
 \begin{align*}
 & \Bigl[\Lambda(t), \widetilde{H}_{\epsilon}(t)\Bigr] = -\Bigl[\Lambda(t), X_{\epsilon}(t)\Bigr]^{\dagger}\widetilde{H}(t)X_{\epsilon}(t) \\
 &\quad  + X_{\epsilon}(t)^{\dagger}\Bigl[\Lambda(t), \widetilde{H}(t)\Bigr]X_{\epsilon}(t) + 
 X_{\epsilon}(t)^{\dagger}\widetilde{H}(t)\Bigl[\Lambda(t), X_{\epsilon}(t)\Bigr]
 \end{align*}
 as in the proof of Lemma 4.1.  Apply Proposition 5.2 and Lemmas 5.3 - 5.4 to the above.  In addition,  apply Proposition 5.2 to $(i\partial \Lambda(t)/\partial t)\Lambda(t)^{-1}$.  Then, we can prove Lemma 5.5 as in the proof of Lemma 4.1.
\end{proof}
	Using Lemma 5.5, we can prove the following as in the proof of Proposition 4.2.
	\begin{pro}
Under Assumption 2.3 there exist functions $q_{a\epsilon}(t,z,\zeta)\ (\adots, \zeo)$ satisfying \eqref{5.18} and 
 \begin{equation} \label{5.31}
Q_{a\epsilon}\tzdz = \Biggl[i\frac{\partial}{\partial t}- \widetilde{H}_{\epsilon}(t), \Lambda(t)^a\Biggr]\Lambda(t)^{-a}.
   \end{equation}
	\end{pro}
		Let $B'^a(\bR^{4d})\ (\adots)$ be the weighted Sobolev spaces introduced in \S 2.  Then we see  that the embedding map from $B'^{a+1}$ into $B'^{a}$ is compact and that  the similar result to Proposition 3.5 follows from Proposition 5.2.  Therefore, using Proposition 5.6, we can prove Theorems 2.4 - 2.6 as in the proofs of Theorems 2.1 - 2.3, respectively.	
 %


\begin{thebibliography}{99} %
%
%
\bibitem{Albeverio et all}
 Albeverio,  S. A.,  H\o egh-Krohn R. J.,  Mazzucchi,  S. (2008).
   {\it Mathematical Theory of Feynman Path Integral. An Introduction, 2nd edition}. Lecture Notes in Math. {\bf 523}.
   Berlin: Springer-Verlag.
%
\bibitem{Berezin et all}
  Berezin,  F. A.,   Shubin, M. A. (1991).
   {\it The Schr\"odinger  Equation}. Dordrecht: Kluwer Academic Publishers.
\bibitem{Cycon et all}
  Cycon, H. L.,  Froese, R. G.,  Kirsch, W.,    Simon, B. (1987).
 {\it Schr\"odinger Operators with Application  to Quantum Mechanics and Global Geometry}. Berlin: Springer-Verlag.
\bibitem{Davies}
  Davies, E. B. (2002).
 Non-self-adjoint differential operators. {\it Bull. London Math. Soc.} {\bf 34}:513-532.
%
\bibitem{Ichinose 1995}
  Ichinose,  W.  (1995).  A note on the existence and $\hbar$-dependency  of the
solution of equations in quantum mechanics. {\it Osaka J. Math.} {\bf 32}:327-345.
%
\bibitem{Ichinose 1999} 
Ichinose,  W.  (1999).  On convergence of the Feynman path integral formulated
through broken line paths. {\it Rev. Math. Phys.}  {\bf 11}:1001-1025. 
%
%
\bibitem{Ichinose 2012} 
Ichinose,  W.  (2012). The continuity and the differentiability of solutions on parameters to the Schr\"odinger equations and the Dirac equations.  {\it J. Pseudo-Differ. Oper. Appl.} {\bf 3}:399-419.
 %
\bibitem{Ichinose 2014}
Ichinose,  W.  (2014). On the Feynman path integral for the Dirac equation in the general dimensional spacetime.
{\it Commun. Math. Phys. } {\bf 329}:483-508.
%
%
\bibitem{Ichinose 2019} 
Ichinose,  W. (2019) On the Feynman path integral for the Schr\"odinger equations with potentials growing polynomially in the spatial direction. ArXiv:1901.05677
%
%
\bibitem{Ichinose 2020} 
Ichinose,  W.  Mathematical theory of  the Feynman path integrals for continuous quantum measurements of positions and spin. In preparation.
%
%
\bibitem{Kumano-go}
 Kumano-go, H. (1981). {\it Pseudo-Differential Operators}. Cambridge: MIT Press.
%
%
%
\bibitem{Leinfelder et all}
Leinfelder, H.,   Simader, C. G. (1981). Schr\"odinger operators with singular magnetic vector potentials. {\it Math. Z.} {\bf 176}:1-19.
%
%
\bibitem{Mensky 1993}
Mensky, M. B. (1993). {\it Continuous Quantum Measurements and Path Integrals}. Bristol: IOP Publishing.%
%
\bibitem{Mensky 2000}
Mensky, M. B. (2000). {\it  Quantum Measurements and Decoherence}. Dordrecht: Kluwer Academic Publishers.%
%
\bibitem{Reed-Simon I}
 Reed, M. ,   Simon,  B. (1980). {\it Methods of Modern Mathematical Physics I: Functional Analysis, Revised and Enlarged Edition}. San Diego: Academic Press.
 %
%
\bibitem{Reed-Simon II}
 Reed, M. ,   Simon,  B. (1975).  {\it Methods of Modern Mathematical Physics II: Fourier Analysis, Self-adjointness}. San Diego: Academic Press. 
 %
 %
\bibitem{Sambou}
Sambou, D. (2014). Lieb-Thirring type inequalities for non-self-adjoint perturbations of magnetic Schr\"odinger operators. {\it J. Funct. Anal.} {\bf 266}:5016-5044.
%
\bibitem{Yajima 1991}
 Yajima, K.  (1991). Schr\"odinger evolution equations with magnetic fields. {\it J. Anal. Math.} {\bf 56}:29-76.
%
\bibitem{Yajima 2011}
Yajima, K.  (2011). Schr\"odinger  equations with time-dependent unbounded singular potentials. {\it  Rev. Math. Phys.}  {\bf 23}:823-838.
%
\bibitem{Yajima 2016}
  Yajima, K.  (2016).  Existence and regularity of propagators for multi-particle Schr\"odinger equations in external fields. {\it Commun. Math. Phys.}  {\bf 347}:103-126.
%
%
%
\end{thebibliography}
\end{document}